\pdfoutput=1
\documentclass[a4paper,11pt]{amsart}

\usepackage[left=1.3in,right=1.3in,top=1in,bottom=1in,includehead,includefoot]{geometry}
\usepackage[utf8]{inputenc}
\usepackage[T1]{fontenc}
\usepackage{lmodern,microtype}
\usepackage{amsmath,amsthm,amssymb,colonequals}
\usepackage{enumitem}
\usepackage{caption}
\usepackage{tikz}
\usetikzlibrary{decorations.pathreplacing}
\usetikzlibrary{shapes.geometric}
\usepackage[vlined,norelsize]{algorithm2e}
\usepackage{hyperref,bookmark}

\hypersetup{colorlinks=true,linkcolor=black,citecolor=black,filecolor=black,urlcolor=black}

\makeatletter
\let\citationorig\citation
\def\citation#1{\citationorig{#1}\@for\@tempa:=#1\do{\@ifundefined{cit@\@tempa}{\global\@namedef{cit@\@tempa}{}}{}}}
\let\bibitemorig\bibitem
\def\bibitem#1{\@ifundefined{cit@#1}{\typeout{LaTeX Warning: Unused bibitem `#1'}}{}\bibitemorig{#1}}
\let\old@setaddresses\@setaddresses
\def\@setaddresses{\bigskip{\parindent 0pt\let\scshape\relax\let\ttfamily\relax\old@setaddresses}}
\makeatother

\def\periodsf{\spacefactor 3000 \space}

\newtheorem{theorem}{Theorem}[section]
\newtheorem{lemma}[theorem]{Lemma}
\newtheorem{proposition}[theorem]{Proposition}

\theoremstyle{definition}

\makeatletter
\newenvironment{thmcontinued}{\begingroup\renewcommand\thmhead[3]{}\@thm{\th@plain\thm@headpunct{}\thm@headsep\z@}{}{}}{\@endtheorem\endgroup}
\newcommand\nobreakpar{\par\nobreak\@afterheading}
\makeatother

\renewenvironment{enumerate}{\begin{enumorig}[label=\textup{(\roman*)}, noitemsep, topsep=3pt plus 3pt, leftmargin=*, widest=iii]}{\end{enumorig}}
\newenvironment{enumeratea}{\begin{enumorig}[label=\textup{(\alph*)}, noitemsep, topsep=3pt plus 3pt, leftmargin=*]}{\end{enumorig}}
\newenvironment{enumeratet}{\leavevmode\nobreakpar\begin{enumorig}[series=theorem, label=\textup{(\arabic*)}, noitemsep, topsep=0pt, labelsep=4pt, leftmargin=*]}{\end{enumorig}}
\newenvironment{enumeratetcontinued}[1]{\begin{enumorig}[resume=theorem, label=\textup{(\arabic*)~(#1)}, ref=\textup{(\arabic*)}, noitemsep, topsep=0pt, labelsep=4pt, leftmargin=*]}{\end{enumorig}}
\newenvironment{enumeratex}[1]{\begin{enumorig}[label=\textup{(#1\arabic*)}, noitemsep, topsep=3pt plus 3pt, leftmargin=*]}{\end{enumorig}}
\newenvironment{enumeratefixed}[1]{\begin{enumorig}[label=\textup{(#1)}, noitemsep, topsep=3pt plus 3pt, leftmargin=*]}{\end{enumorig}}

\renewenvironment{itemize}{\begin{itemorig}[label=\textbullet, noitemsep, topsep=3pt plus 3pt, leftmargin=1.1em]}{\end{itemorig}}

\DeclareMathOperator{\dom}{dom}

\makeatletter
\newcommand{\superimpose}[2]{\ooalign{$#1\@firstoftwo#2$\cr\hfil$#1\@secondoftwo#2$\hfil\cr}}
\makeatother

\let\leq\leqslant
\let\geq\geqslant
\let\setminus\smallsetminus
\let\coloneq\colonequals
\let\Theta\varTheta
\let\Omega\varOmega

\def\includes{\mathrel{\mathpalette\superimpose{{\supset}{\circ}}}}
\let\overlaps\between

\def\setN{\mathbb{N}}
\def\setR{\mathbb{R}}
\def\setZ{\mathbb{Z}}

\def\calC{\mathcal{C}}
\def\calF{\mathcal{F}}
\def\calG{\mathcal{G}}
\def\calI{\mathcal{I}}

\def\gamABS{\mathsf{CABS}}
\def\gamCOCO{\mathsf{COCO}}
\def\gamG{\mathsf{G}}
\def\gamINT{\mathsf{INT}}
\def\gamIOV{\mathsf{CIOV}}

\def\leftside{\ell}
\def\rightside{r}

\allowdisplaybreaks

\hypersetup{
  pdftitle={On-line approach to off-line coloring problems on graphs with geometric representations},
  pdfauthor={Tomasz Krawczyk, Bartosz Walczak}
}

\title[On-line approach to off-line coloring problems]{On-line approach to off-line coloring problems\\on graphs with geometric representations}

\author{Tomasz Krawczyk\and Bartosz Walczak}

\address{Department of Theoretical Computer Science, Faculty of Mathematics and Computer Science, Jagiellonian University, Kraków, Poland}
\email{\href{mailto:krawczyk@tcs.uj.edu.pl}{krawczyk@tcs.uj.edu.pl}, \href{mailto:walczak@tcs.uj.edu.pl}{walczak@tcs.uj.edu.pl}}

\thanks{A journal version of this paper appeared in \href{http://doi.org/10.1007/s00493-016-3414-x}{\emph{Combinatorica}, in press}.}
\thanks{A preliminary version of this paper appeared as: \href{http://doi.org/10.1007/978-3-662-43948-7_61}{Coloring relatives of interval overlap graphs via on-line games, in: Javier Esparza, Pierre Fraigniaud, Thore Husfeldt, and Elias Koutsoupias (eds.), \emph{41st International Colloquium on Automata, Languages, and Programming (ICALP 2014)}, part~I, vol.~8572 of \emph{Lecture Notes Comput.\ Sci.}, pp.~738--750, Springer, Berlin, 2014}.}
\thanks{Tomasz Krawczyk and Bartosz Walczak were partially supported by National Science Center of Poland grant 2011/03/B/ST6/01367.
Bartosz Walczak was partially supported by Swiss National Science Foundation grant 200020-144531.}

\begin{document}

\begin{abstract}
The main goal of this paper is to formalize and explore a connection between chromatic properties of graphs with geometric representations and competitive analysis of on-line algorithms, which became apparent after the recent construction of triangle-free geometric intersection graphs with arbitrarily large chromatic number due to Pawlik et~al.
We show that on-line graph coloring problems give rise to classes of \emph{game graphs} with a natural geometric interpretation.
We use this concept to estimate the chromatic number of graphs with geometric representations by finding, for appropriate simpler graphs, on-line coloring algorithms using few colors or proving that no such algorithms exist.

We derive upper and lower bounds on the maximum chromatic number that rectangle overlap graphs, subtree overlap graphs, and interval filament graphs (all of which generalize interval overlap graphs) can have when their clique number is bounded.
The bounds are absolute for interval filament graphs and asymptotic of the form $(\log\log n)^{\smash{f(\omega)}}$ for rectangle and subtree overlap graphs, where $f(\omega)$ is a polynomial function of the clique number and $n$ is the number of vertices.
In particular, we provide the first construction of geometric intersection graphs with bounded clique number and with chromatic number asymptotically greater than $\log\log n$.

We also introduce a concept of $K_k$-free colorings and show that for some geometric representations, $K_3$-free chromatic number can be bounded in terms of clique number although the ordinary ($K_2$-free) chromatic number cannot.
Such a result for segment intersection graphs would imply a well-known conjecture that $k$-quasi-planar geometric graphs have linearly many edges.
\end{abstract}

\maketitle

\section{Introduction}

Graphs represented by geometric objects have been attracting researchers for many reasons, ranging from purely aesthetic to practical ones.
A problem which has been extensively studied for this kind of graphs is proper coloring: given a family of objects, one wants to color them with few colors so that any two objects generating an edge of the graph obtain distinct colors.
The off-line variant of the problem, in which the entire graph to be colored is known in advance, finds practical applications in areas like channel assignment, map labeling, and VLSI design.
The on-line variant, in which the graph is being revealed piece by piece and the coloring agent must make irrevocable decisions without knowledge of the entire graph, is a common model for many scheduling problems.
A natural connection between the two variants, which is discussed in this paper, allows us to establish new bounds on the chromatic number in various classes of graphs by analyzing the on-line problem in much simpler classes of graphs.

We write $\chi$, $\omega$ and $n$ to denote the chromatic number, the clique number (maximum size of a clique), and the number of vertices of a graph under consideration, respectively.
If $\chi=\omega$ holds for a graph $G$ and all its induced subgraphs, then $G$ is \emph{perfect}.
A class of graphs $\calG$ is \emph{$\chi$-bounded} or \emph{near-perfect} if there is a function $f\colon\setN\to\setN$ such that every graph in $\calG$ satisfies $\chi\leq f(\omega)$.
All graphs that we consider are finite.

\subsection*{Geometric intersection and overlap graphs}

Any finite family of sets $\calF$ gives rise to two graphs with vertex set $\calF$: the \emph{intersection graph}, whose edges connect pairs of intersecting members of $\calF$, and the \emph{overlap graph}, whose edges connect pairs of members of $\calF$ that \emph{overlap}, that is, intersect but are not nested.
In this paper, we do not want to distinguish isomorphic graphs, and hence we call a graph $G$ an \emph{intersection}/\emph{overlap graph} of $\calF$ if there is a bijective mapping $\mu\colon V(G)\to\calF$ such that $uv\in E(G)$ if and only if $\mu(u)$ and $\mu(v)$ intersect/overlap.
Depending on the context, we call the mapping $\mu$ or the family $\calF$ an \emph{intersection}/\emph{overlap model} or \emph{representation} of $G$.
Ranging over all representations of a particular kind, for example, by sets with a specific geometric shape, we obtain various classes of intersection and overlap graphs.
Prototypical examples are \emph{interval graphs} and \emph{interval overlap graphs}, which are intersection and overlap graphs, respectively, of closed intervals in $\setR$.
Interval overlap graphs are the same as \emph{circle graphs}---intersection graphs of chords of a circle.

Interval graphs are well known to be perfect.
Interval overlap graphs are no longer perfect, but they are near-perfect: Gyárfás \cite{Gya85} proved that every interval overlap graph satisfies $\chi=O(\omega^24^\omega)$, which was improved to $\chi=O(\omega^22^\omega)$ by Kostochka \cite{Kos88}, and further to $\chi=O(2^\omega)$ by Kostochka and Kratochvíl \cite{KK97} (specifically, they proved $\chi\leq 50\cdot 2^\omega-32\omega-64$, which was later improved to $\chi\leq 21\cdot 2^\omega-24\omega-24$ by Černý \cite{Cer07}).
Currently the best lower bound on the maximum chromatic number of an interval overlap graph with clique number $\omega$ is $\Omega(\omega\log\omega)$, due to Kostochka \cite{Kos88}.
The exponential gap between the best known upper and lower bounds remains open for over 30 years.
For triangle-free interval overlap graphs, the bound is $\chi\leq 5$ \cite{Kos88}, and it is tight \cite{Age96}.

An overlap model is \emph{clean} if it has no three sets such that two overlapping ones both contain the third one.
An overlap graph is \emph{clean} if it has a clean overlap model.
Clean overlap graphs are often much easier to color than general (non-clean) ones.
For example, Kostochka and Milans \cite{KM12} proved that clean interval overlap graphs satisfy $\chi\leq 2\omega-1$.
In this paper, several upper bounds on the chromatic number are proved first for clean overlap graphs and then (with weaker bounds) for general overlap graphs.

Intervals in $\setR$ are naturally generalized by axis-parallel rectangles in $\setR^2$ and by subtrees of a tree, which give rise to the following classes of graphs:
\begin{itemize}
\item \emph{chordal graphs}---intersection graphs of subtrees of a tree, originally defined as graphs with no induced cycles of length greater than $3$, see \cite{Gav74},
\item \emph{subtree overlap graphs}---overlap graphs of subtrees of a tree, introduced in \cite{Gav00},
\item \emph{rectangle graphs}---intersection graphs of axis-parallel rectangles in the plane,
\item \emph{rectangle overlap graphs}---overlap graphs of axis-parallel rectangles in the plane.
\end{itemize}

Chordal graphs are perfect.
Rectangle graphs are near-perfect: Asplund and Grün\-baum \cite{AG60} proved that every rectangle graph satisfies $\chi=O(\omega^2)$ (specifically, they proved $\chi\leq 4\omega^2-3\omega$, which was later improved to $\chi\leq 3\omega^2-2\omega-1$ by Hendler \cite{Hen98}).
Kostochka \cite{Kos04} claimed existence of rectangle graphs with chromatic number $3\omega$, and no better construction is known.

Rectangle overlap graphs are no longer near-perfect: Pawlik et~al.\ \cite{PKK+13} presented a construction of triangle-free rectangle overlap graphs with chromatic number $\Theta(\log\log n)$.
This construction works also for a variety of other geometric intersection graphs \cite{PKK+13,PKK+14} and is used in all known counterexamples to a conjecture of Scott on graphs with an excluded induced subdivision \cite{CELOdM16}.
Actually, it produces graphs that we call \emph{interval overlap game graphs}, which form a subclass of rectangle overlap graphs, segment intersection graphs, and subtree overlap graphs.
This implies that subtree overlap graphs are not near-perfect either.
Interval overlap game graphs play an important role in this paper, but their definition requires some preparation, so it is postponed until Section \ref{sec:rectangle-subtree}.
It is proved in \cite{KPW15} that triangle-free rectangle overlap graphs have chromatic number $O(\log\log n)$, which matches the above-mentioned lower bound.
It is worth noting that intersection graphs of axis-parallel boxes in $\setR^3$ are not near-perfect either: Burling \cite{Bur65} constructed such graphs with no triangles and with chromatic number $\Theta(\log\log n)$.
We reprove Burling's result in Section \ref{sec:examples}.

\emph{Interval filament graphs} are intersection graphs of \emph{interval filaments}, which are continuous non-negative functions defined on closed intervals with value zero on the endpoints.
Interval filament graphs were introduced in \cite{Gav00} as a generalization of interval overlap graphs, polygon-circle graphs, chordal graphs and co-comparability graphs.
Every interval filament graph is a subtree overlap graph \cite{ES07}, and the overlap graph of any collection of subtrees of a tree $T$ intersecting a common path in $T$ is an interval filament graph \cite{ES07}.
We comment more on this in Section \ref{sec:rectangle-subtree}.
An interval filament graph is \emph{domain-non-overlapping} if it has an intersection representation by interval filaments whose domains are pairwise non-overlapping intervals.

\emph{Outerstring graphs} are intersection graphs of curves in a halfplane with one endpoint on the boundary of the halfplane.
Every interval filament graph is an outerstring graph.

\emph{String graphs} are intersection graphs of arbitrary curves in the plane.
Every graph of any class considered above is a string graph.
For example, a rectangle overlap graph can be represented as an intersection graph of boundaries of rectangles, and a subtree overlap graph defined by subtrees of a tree $T$ can be represented as the intersection graph of closed curves encompassing these subtrees in a planar drawing of $T$.
The best known upper bound on the chromatic number of string graphs is $(\log n)^{\smash{O(\log\omega)}}$, due to Fox and Pach \cite{FP14}.

The following diagram illustrates the inclusions between most of the classes defined above:

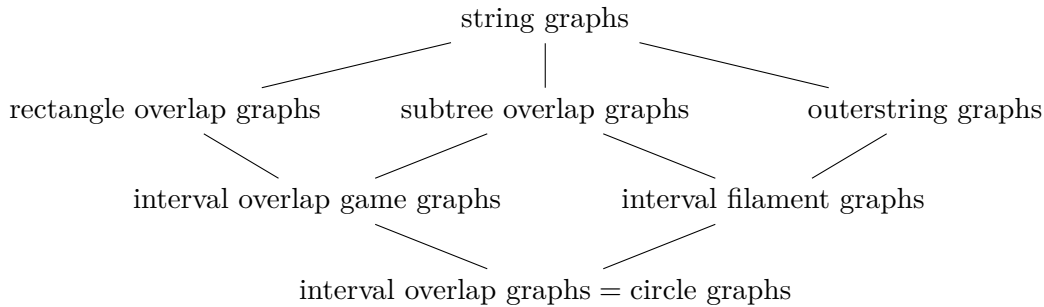
\begin{center}
\begin{tikzpicture}[xscale=2.5,yscale=1.2]
\node (intov) at (0,0) {interval overlap graphs${}={}$circle graphs};
\node (game) at (-1.2,1) {interval overlap game graphs};
\node (intfil) at (1.2,1) {interval filament graphs};
\node (rectov) at (-2,2) {rectangle overlap graphs};
\node (subov) at (0,2) {subtree overlap graphs};
\node (outer) at (2,2) {outerstring graphs};
\node (string) at (0,3) {string graphs};
\path (intov) edge (game) edge (intfil);
\path (game) edge (rectov) edge (subov);
\path (intfil) edge (subov) edge (outer);
\path (string) edge (rectov) edge (subov) edge (outer);
\end{tikzpicture}
\end{center}

\subsection*{Results}

Here is the summary of the results of this paper.
In what follows, we write $O_\omega$ and $\Theta_\omega$ to denote asymptotics with $\omega$ fixed as a constant.

\begin{theorem}
\label{thm:filament}
\begin{enumeratet}
\item\label{item:filament-upper} Every interval filament graph satisfies\/ $\chi\leq g(\omega)\binom{\omega+1}{2}$, where\/ $g(\omega)$ denotes the upper bound on the chromatic number of interval overlap graphs with clique number\/ $\omega$.
\item\label{item:non-overlapping-filament-upper} Every domain-non-overlapping interval filament graph satisfies\/ $\chi\leq\binom{\omega+1}{2}$.
\item\label{item:non-overlapping-filament-lower} There are domain-non-overlapping interval filament graphs with\/ $\chi=\binom{\omega+1}{2}$.
\end{enumeratet}
\end{theorem}

\begin{theorem}
\label{thm:subtree}
\begin{enumeratet}
\item\label{item:subtree-upper} Every subtree overlap graph satisfies\/ $\chi=O_\omega((\log\log n)^{\smash{\binom{\omega}{2}}})$.
\item\label{item:clean-subtree-upper} Every clean subtree overlap graph satisfies\/ $\chi=O_\omega((\log\log n)^{\omega-1})$.
\item\label{item:clean-subtree-lower} There are clean subtree overlap graphs with\/ $\chi=\Theta_\omega((\log\log n)^{\omega-1})$.
Consequently, there are string graphs with\/ $\chi=\Theta_\omega((\log\log n)^{\omega-1})$.
\end{enumeratet}
\end{theorem}

\begin{theorem}
\label{thm:rectangle}
\begin{enumeratet}
\item\label{item:rectangle-upper} Every rectangle overlap graph satisfies\/ $\chi=O_\omega((\log\log n)^{\omega-1})$.
\item\label{item:clean-rectangle-upper} Every clean rectangle overlap graph satisfies\/ $\chi=O_\omega(\log\log n)$.
\end{enumeratet}
\end{theorem}

\noindent
The aforementioned result of Pawlik et~al.\ \cite{PKK+13} complements Theorem \ref{thm:rectangle} in line with Theorems \ref{thm:filament} and \ref{thm:subtree}, showing that its statement \ref{item:clean-rectangle-upper} is asymptotically tight:

\begin{thmcontinued}
\begin{enumeratetcontinued}{\cite{PKK+13}}
\item\label{item:clean-rectangle-lower} There are clean rectangle overlap graphs with\/ $\omega=2$ and\/ $\chi=\Theta(\log\log n)$.
\end{enumeratetcontinued}
\end{thmcontinued}

The special case of Theorem \ref{thm:rectangle} \ref{item:rectangle-upper}--\ref{item:clean-rectangle-upper} for $\omega=2$ was proved in \cite{KPW15}.
The Pawlik et~al.\ result \ref{item:clean-rectangle-lower} above implies Theorem \ref{thm:subtree}~\ref{item:clean-subtree-lower} for $\omega=2$, which we comment on in Section \ref{sec:rectangle-subtree}.
Theorem \ref{thm:subtree}~\ref{item:clean-subtree-lower} provides the first construction of string graphs with bounded clique number and with chromatic number asymptotically greater than $\log\log n$.

Theorem \ref{thm:filament}~\ref{item:filament-upper} asserts in particular that the class of interval filament graphs is $\chi$-bounded.
This is also implied by a recent result of Rok and Walczak \cite{RW14} that the class of outerstring graphs is $\chi$-bounded, which is proved using different techniques leading to an enormous bound on the chromatic number.
Here, by contrast, the bound is pretty good.
For instance, it follows that triangle-free interval filament graphs have chromatic number at most $15$.
The fact that the class of interval filament graphs is $\chi$-bounded implies that it is a \emph{proper} subclass of the class of subtree overlap graphs, as the latter is not $\chi$-bounded.
However, we are not aware of any reasonably small graph witnessing proper inclusion between these two classes.
The example resulting from the bounds on the chromatic number (for $\omega=2$) has more than $2^{2^{14}}$ vertices.

A \emph{$K_k$-free coloring} of a graph $G$ is a coloring of the vertices of $G$ such that every color class induces a $K_k$-free subgraph of $G$.
A $K_2$-free coloring is just a proper coloring.
The \emph{$K_k$-free chromatic number}, denoted by $\chi_k$, is the minimum number of colors sufficient for a $K_k$-free coloring of the graph.
Our interest in $K_k$-free colorings comes from an attempt to prove the so-called \emph{quasi-planar graph conjecture}, which is discussed at the end of this section.
The proof of Theorem \ref{thm:rectangle}~\ref{item:clean-rectangle-upper} gives the following as a by-product.

\begin{theorem}
\label{thm:rectangle-K3}
Every clean rectangle overlap graph satisfies\/ $\chi_3=O_\omega(1)$.
\end{theorem}

On the other hand, Theorem \ref{thm:subtree} \ref{item:clean-subtree-upper}--\ref{item:clean-subtree-lower} implies that for every $k\geq 2$, there are clean subtree overlap graphs (and thus string graphs) with $\omega=k$ and $\chi_k=\Theta_k(\log\log n)$.
To see this, consider any $K_k$-free coloring of a clean subtree overlap graph with $\omega=k$ and $\chi=\Theta_k((\log\log n)^{k-1})$ guaranteed by Theorem \ref{thm:subtree}~\ref{item:clean-subtree-lower}.
Every color class induces a clean subtree overlap graph with $\omega\leq k-1$ and therefore, by Theorem \ref{thm:subtree}~\ref{item:clean-subtree-upper}, with $\chi=O_k((\log\log n)^{k-2})$.
Hence, there must be at least $\Theta_k(\log\log n)$ color classes.

The proofs of the upper bounds in Theorems \ref{thm:filament}--\ref{thm:rectangle-K3} are constructive---they can be used to design polynomial-time coloring algorithms that use the claimed number of colors.
These algorithms require that the input graph is provided together with its geometric representation.
Constructing a representation is at least as hard as deciding whether a representation exists (the recognition problem), which is NP-complete for interval filament graphs \cite{Per07}, and whose complexity is unknown for subtree overlap graphs and rectangle overlap graphs.

\subsection*{Methods}

All our proofs heavily depend on a correspondence between on-line graph coloring problems and off-line colorings of so-called \emph{game graphs}, which originates from considerations in \cite{KPW15,PKK+13} and which we formalize in the next section.
It allows us to reduce problems of estimating the maximum possible chromatic number in classes of geometric intersection graphs to designing coloring algorithms or adversary strategies for the on-line coloring problem in much simpler classes of graphs.
For classes of geometric intersection graphs with bounded clique number and unbounded chromatic number, this is the only approach known to give upper bounds on the chromatic number better than single logarithmic (with respect to $n$).

In Section \ref{sec:examples}, we illustrate the concept of game graphs on two short examples.
First, we construct rectangle graphs with chromatic number $3\omega-2$, which is only less by $2$ than Kostochka's claimed but unpublished lower bound of $3\omega$.
Second, we reproduce Burling's construction of triangle-free intersection graphs of axis-parallel boxes in $\setR^3$ with $\chi=\Theta(\log\log n)$.
Later sections contain the proofs of Theorems \ref{thm:filament}--\ref{thm:rectangle-K3}.

The proof of Theorem \ref{thm:filament} relies on a result of Felsner \cite{Fel97}, which determines precisely the competitiveness of the on-line coloring problem on incomparability graphs of up-growing partial orders.
The proofs of Theorems \ref{thm:subtree} and \ref{thm:rectangle} rely on the coloring algorithm and the adversary strategy for the on-line coloring problem on forests.
A well-known adversary strategy due to Bean \cite{Bea76}, later rediscovered by Gyárfás and Lehel \cite{GL88}, forces any on-line coloring algorithm to use at least $c$ colors on a forest with at most $2^{c-1}$ vertices.
This is tightly matched by the algorithm called First-fit, discussed in Section \ref{sec:coloring}, which colors every $n$-vertex forest on-line using at most $\lfloor\log_2n\rfloor+1$ colors.
A reduction to on-line coloring of forests is a final step in the proofs of Theorem \ref{thm:subtree}~\ref{item:clean-subtree-upper} and Theorem \ref{thm:rectangle}~\ref{item:clean-subtree-upper}.
Bean's adversary strategy underlies the results of \cite{PKK+13,PKK+14}, in particular, Theorem \ref{thm:rectangle}~\ref{item:clean-rectangle-lower}, whereas a generalization of Bean's strategy, which is presented in Section \ref{sec:construction}, underlies Theorem \ref{thm:subtree}~\ref{item:clean-subtree-lower}.

An important ingredient in the proofs of Theorem \ref{thm:subtree}~\ref{item:subtree-upper} and Theorem \ref{thm:rectangle}~\ref{item:rectangle-upper} is a generalized breadth-first search procedure, which we call \emph{$k$-clique breadth-first search} and which may be of independent interest.
It allows us to reduce the respective coloring problem to clean overlap graphs in a similar way as the ordinary breadth-first search does when $\omega=2$ \cite{Gya85,KPW15}.
This is discussed in detail in Section \ref{sec:reduction-to-clean}.

\subsection*{Problems}

The following problem, posed in \cite{PKK+14}, remains open: estimate (asymptotically with respect to $n$) the maximum possible chromatic number for triangle-free segment intersection graphs or, more generally, segment intersection graphs with bounded clique number.
We believe the answer is $O_\omega((\log\log n)^c)$ for some constant $c\geq 1$.
For the analogous problem for string graphs, we believe the answer is $O_\omega((\log\log n)^{\smash{f(\omega)}})$ for some function $f\colon\setN\to\setN$ with $f(\omega)\geq\omega-1$.
The first step of the proof of Theorem \ref{thm:rectangle}~\ref{item:clean-rectangle-upper} is a reduction from clean rectangle overlap graphs to interval overlap game graphs (see Lemma \ref{lem:rectangle-to-game}).
The main challenge in applying the on-line approach to the problems above lies in devising an analogous reduction from segment or string graphs to game graphs of an appropriate on-line graph coloring problem.

An exciting open problem related to geometric intersection graphs concerns the number of edges in $k$-quasi-planar graphs.
A graph drawn in the plane is \emph{$k$-quasi-planar} if no $k$ edges cross each other in the drawing.
Pach, Shahrokhi and Szegedy \cite{PSS96} conjectured that $k$-quasi-planar graphs have $O_k(n)$ edges.
For $k=2$, this asserts the well-known fact that planar graphs have $O(n)$ edges.
The conjecture is also proved for $k=3$ \cite{AAP+97,PRT06} and $k=4$ \cite{Ack09}, but it remains open for $k\geq 5$.
The best known upper bounds on the number of edges in a $k$-quasi-planar graph with $k\geq 5$ are $n(\log n)^{O\smash[t]{(\log k)}}$ in general \cite{FP12,FP14} and $O_k(n\log n)$ if the edges are drawn as straight-line segments \cite{Val98} or $1$-intersecting curves \cite{SW15}.
If we can prove that the intersection graph of the edges of a $k$-quasi-planar graph $G$ satisfies $\chi_3=O_k(1)$ (or $\chi_4=O_k(1)$), then it will follow that $G$ has $O_k(n)$ edges, as each color class in a $K_3$-free ($K_4$-free) coloring of the edges of $G$ is itself a $3$-quasi-planar ($4$-quasi-planar) graph and therefore has $O(n)$ edges.
The construction of triangle-free segment intersection graphs with arbitrarily large chromatic number \cite{PKK+14} implies that such an approach cannot succeed when we ask for a proper coloring of the edges instead of a $K_3$-free ($K_4$-free) coloring.
In view of the remark after Theorem \ref{thm:rectangle-K3}, neither can it succeed for $K_k$-free colorings when the edges of $G$ are allowed to cross arbitrarily many times.
Nevertheless, Theorem \ref{thm:rectangle-K3} suggests a substantial difference between proper and triangle-free colorings of geometric intersection graphs, which makes this approach appealing for $k$-quasi-planar graphs whose edges are drawn as straight-line segments or, more generally, $1$-intersecting curves.

Finally, an interesting challenge is to close the asymptotic gap between the upper bounds of $O_\omega((\log\log n)^{\smash{\binom{\omega}{2}}})$ and $O_\omega((\log\log n)^{\omega-1})$ and the lower bounds of $\Omega_\omega((\log\log n)^{\omega-1})$ and $\Omega(\log\log n)$, respectively, on the maximum chromatic number of subtree and rectangle overlap graphs.
We believe that the lower bounds are correct.
A problem of similar flavor is to prove the analogue of Theorem \ref{thm:rectangle-K3} for rectangle overlap graphs that are not clean.

\section{On-line graph coloring games and game graphs}
\label{sec:games}

The \emph{on-line graph coloring game} is played by two deterministic players: Presenter and Algorithm.
It is played in rounds.
In each round, Presenter adds a new vertex to the graph and declares whether or not it has an edge to each of the vertices presented before.
As a response, in the same round, Algorithm colors this vertex keeping the property that the coloring is proper.
Imposing additional restrictions on Presenter's moves gives rise to many possible variants of the on-line graph coloring game.
Typical kinds of such restrictions look as follows:
\begin{enumerate}
\item The graph $G$ being built by Presenter keeps belonging to a specific class of graphs $\calG$.
It is reasonable to require that the class $\calG$ is hereditary (closed under taking induced subgraphs).
\item In addition to $G$, Presenter builds a mapping $\mu\colon V(G)\to\calC$ called a \emph{representation} of $G$ in some class of objects $\calC$, and the edges of $G$ are defined in terms of $\mu$.
\item In addition to $G$, Presenter builds relations $R_1,\ldots,R_r$ on $V(G)$, and the edges of $G$ are defined in terms of $R_1,\ldots,R_r$.
\item There can be some further restrictions relating $\mu$, $R_1,\ldots,R_r$, and the order in which the vertices are presented.
\end{enumerate}
The final graph to be built by Presenter is not fixed in advance and can depend on the decisions taken by Algorithm when coloring vertices.
However, the decisions of both players are irrevocable: Presenter cannot change the part of the graph, the representation, or the relations after they have been presented, and Algorithm cannot change the colors after they have been assigned.
The goal of Algorithm is to keep using as few colors as possible, while Presenter wants to force Algorithm to use as many colors as possible.
The \emph{value} of such a game is the minimum number $c$ such that Algorithm has a strategy to color any graph that can be presented in the game using at most $c$ colors or, equivalently, the maximum number $c$ such that Presenter has a strategy to force Algorithm to use at least $c$ colors regardless of how Algorithm responds.

We call any variant of the on-line graph coloring game simply an \emph{on-line game}, and any coloring strategy of Algorithm simply an \emph{on-line algorithm}.
We let $\prec$ denote the order in which the vertices are presented.
It is envisioned as going from left to right.

Now, we explain the crucial concept of our paper---game graphs.
Let $\gamG$ be an on-line game with representation $\mu$ in a class $\calC$ and with relations $R_1,\ldots,R_r$.
If a graph $G$, a particular representation $\mu\colon V(G)\to\calC$, and particular relations $R_1,\ldots,R_r$ on $V(G)$ are allowed to be presented in $n$ rounds of the game $\gamG$ in such a way that the vertices are presented in a particular order $\prec$ on $V(G)$, then we call the tuple $G,\mu,R_1,\ldots,R_r,\prec$ an \emph{$n$-round presentation scenario} in $\gamG$.
We define the class of \emph{game graphs} associated with $\gamG$ as follows.
A graph $G$ is a \emph{game graph} of the on-line game $\gamG$ if there exist a rooted forest $F$ on $V(G)$, a mapping $\mu\colon V(G)\to\calC$, and relations $R_1,\ldots,R_r$ on $V(G)$ such that
\begin{enumeratea}
\item\label{item:game-graph-i} for every $v\in V(G)$, the subgraph $G[V(P_v)]$ of $G$ induced on the vertices of the path $P_v$ in $F$ from a root to $v$, the representation $\mu$ restricted to $V(P_v)$, the relations $R_1,\ldots,R_r$ restricted to $V(P_v)$, and the order $\prec$ of vertices along $P_v$ form a valid $|V(P_v)|$-round presentation scenario in $\gamG$,
\item\label{item:game-graph-ii} if $uv\in E(G)$, then $u$ is an ancestor of $v$ or $v$ is an ancestor of $u$ in $F$.
\end{enumeratea}
For any two distinct vertices $u$ and $v$ of a game graph, we let $u\prec v$ denote that $u$ is an ancestor of $v$ in $F$.
Therefore, the order of presentation $\prec$ in the on-line game and the relation $\prec$ in the game graph correspond to each other in the same way as the relations $R_1,\ldots,R_r$ do in the on-line game and in the game graph.
A game graph can be envisioned as a union of several presentation scenarios in which some (but not necessarily all) common prefixes of these scenarios have been identified.

All the games that we will consider are closed under taking induced subgraphs, in the sense that any induced subgraph of any presentation scenario (where the representation, the relations, and the order $\prec$ are restricted to the vertices of the subgraph) is again a valid presentation scenario.
It easily follows from the definition that the game graphs of such games are also closed under taking induced subgraphs.

It follows from \ref{item:game-graph-ii} that $\omega(G)=\max\{\omega(G[V(P_v)])\colon v\in V(G)\}$.
In particular, if one of the restrictions on the game $\gamG$ requires that the presented graph has clique number at most $k$, then all game graphs of $\gamG$ also have clique number at most $k$.

\begin{lemma}
\label{lem:game-upper}
If there is an on-line algorithm using at most\/ $c$ colors in an on-line game\/ $\gamG$, then every game graph of\/ $\gamG$ has chromatic number at most\/ $c$.
\end{lemma}

\begin{proof}
Intuitively, to color a game graph properly, it is enough to run the on-line algorithm separately on the subgraph induced on each path in $F$ from a root to a leaf.

More formally, let $G$ be a game graph of $\gamG$ with underlying forest $F$, representation $\mu$, and relations $R_1,\ldots,R_r$.
For every $u\in V(G)$, the condition \ref{item:game-graph-i} of the definition of a game graph gives us a presentation scenario of the graph $G[V(P_u)]$.
Color the vertex $u$ in $G$ with the color assigned to $u$ by Algorithm in this scenario.
For every descendant $v$ of $u$ in $F$, the presentation scenario of $G[V(P_u)]$ is the initial part of the presentation scenario of $G[V(P_v)]$ up to the point when $u$ is presented, so Algorithm assigns the same color to $u$ in both scenarios.
Therefore, since Algorithm colors every $G[V(P_v)]$ properly, the coloring of $G$ defined this way is also proper.
\end{proof}

We say that a strategy of Presenter in an on-line game $\gamG$ is \emph{finite} if the total number of presentation scenarios that can occur in the game when Presenter plays according to this strategy, for all possible responses of Algorithm, is finite.

\begin{lemma}
\label{lem:game-lower}
If Presenter has a finite strategy to force Algorithm to use at least\/ $c$ colors in an on-line game\/ $\gamG$, then there exists a game graph of\/ $\gamG$ with chromatic number at least\/ $c$.
Moreover, the number of vertices of this graph is equal to the total number of presentation scenarios that can occur with this strategy.
\end{lemma}

\begin{proof}
Consider a finite strategy of Presenter forcing Algorithm to use at least $k$ colors in $\gamG$.
Let $S$ be the set of presentation scenarios that can occur when Presenter plays according to this strategy.
Hence, $S$ is finite.
Define a forest $F$ on $S$ so that
\begin{itemize}
\item if $s\in S$ is a scenario that presents only one vertex, then $s$ is a root of $F$,
\item otherwise, the parent of $s$ in $F$ is the scenario with one vertex less, describing the situation in the game before the last vertex is presented in the scenario $s$.
\end{itemize}
For a scenario $s\in S$, let $v(s)$ denote the last vertex presented in the scenario $s$.
We define a graph $G$ on $S$ so that $s_1s_2$ is an edge of $G$ if $s_1$ is an ancestor of $s_2$ and $v(s_1)v(s_2)$ is an edge in the graph presented in the scenario $s_2$ or vice versa.
We define relations $R_1,\ldots,R_r$ on $S$ in the same way: $s_1\mathrel{R_i}s_2$ if $s_1$ is an ancestor of $s_2$ and $v(s_1)\mathrel{R_i}v(s_2)$ in the scenario $s_2$ or vice versa.
Finally, for $s\in S$, we define $\mu(s)=\mu(v(s))$ in the scenario $s$.
It clearly follows that the graph $G$ thus obtained is a game graph of $\gamG$ with underlying forest $F$, representation $\mu$, and relations $R_1,\ldots,R_r$.

It remains to prove that $\chi(G)\geq c$.
Suppose to the contrary that there is a proper coloring of $G$ using $c-1$ colors.
Consider the following strategy of Algorithm against Presenter's considered strategy in $\gamG$.
When a new vertex is presented, Algorithm looks at the presentation scenario $s$ of the structure presented so far.
Since Presenter is assumed to play according to the strategy that gives rise to the game graph $G$, the scenario $s$ is a vertex of $G$.
Algorithm colors the new vertex $v(s)$ in the game with the color of $s$ in the assumed coloring of $G$ using $c-1$ colors.
This way, Algorithm uses only $c-1$ colors against Presenter's considered strategy, which contradicts the assumption that this strategy forces Algorithm to use at least $c$ colors.
\end{proof}

Here is how Lemmas \ref{lem:game-upper} and \ref{lem:game-lower} are typically applied.
To provide an upper bound on the chromatic number of graphs of some class $\calG$, we show that each graph in $\calG$ is a game graph of an appropriately chosen on-line game, and we find an on-line algorithm for this game using few colors.
To construct graphs of some class $\calG$ with large chromatic number, we show that every game graph of an appropriately chosen on-line game belongs to $\calG$, and we find a finite strategy of Presenter in this game forcing Algorithm to use many colors.
We use this approach to prove the results of the paper.
First, we reduce Theorems \ref{thm:filament}--\ref{thm:rectangle-K3} to claims about game graphs of appropriately chosen on-line games.
Then, to prove these claims, we devise strategies for Algorithm or Presenter in these games and apply Lemmas \ref{lem:game-upper} and \ref{lem:game-lower} accordingly.

\section{Two simple examples}
\label{sec:examples}

To illustrate the concept developed in the previous section, we prove the following.

\begin{proposition}
\label{prop:rectangle}
There are rectangle graphs with chromatic number\/ $3\omega-2$.
\end{proposition}

Let $\calI$ denote the set of all closed intervals in $\setR$.
Consider an on-line game $\gamINT(k)$ on the class of interval graphs with clique number at most $k$ presented with their interval representation.
That is, Presenter builds an interval graph $G$ and a representation $\mu\colon V(G)\to\calI$ so that
\begin{enumerate}
\item $\mu$ is an intersection model of $G$, that is, $uv\in E(G)$ if and only if $\mu(u)\cap\mu(v)\neq\emptyset$,
\item $\omega(G)\leq k$,
\end{enumerate}
and Algorithm properly colors $G$ on-line.
For this game, the definition of a game graph comes down to the following: a graph $G$ is a game graph of $\gamINT(k)$ if there exist a rooted forest $F$ on $V(G)$ and a mapping $\mu\colon V(G)\to\calI$ such that
\begin{enumeratea}
\item for every $v\in V(G)$ and for the path $P_v$ in $F$ from a root to $v$, the following holds:
\begin{enumerate}
\item $\mu$ restricted to $V(P_v)$ is an intersection model of $G[V(P_v)]$,
\item $\omega(G[V(P_v)])\leq k$,
\end{enumerate}
\item if $uv\in E(G)$, then $u$ is an ancestor of $v$ or $v$ is an ancestor of $u$ in $F$.
\end{enumeratea}
Recall that the ancestor-descendant order of $F$ is denoted by $\prec$.
The above can be simplified to the following two conditions, which correspond to the two conditions in the definition of the game $\gamINT(k)$:
\begin{enumerate}
\item $uv\in E(G)$ if and only if $u\prec v$ or $v\prec u$ and $\mu(u)\cap\mu(v)\neq\emptyset$,
\item $\omega(G)\leq k$.
\end{enumerate}

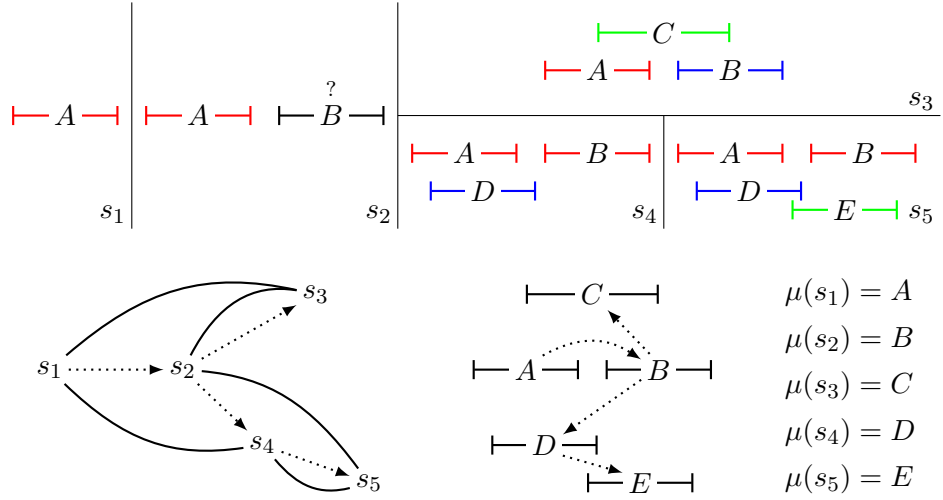
\begin{figure}[t]
\centering
\begin{tikzpicture}[xscale=0.7,yscale=0.5]
  \tikzstyle{every node}=[inner sep=2pt,fill=white]
  \draw[thick,red,|-|] (0.75,0)--(2.75,0);
  \node at (1.75,0) {$A$};
  \node[above left] at (3,-3) {$s_1$};
  \draw (3,-3)--(3,3);
  \begin{scope}[shift={(3,0)}]
    \draw[thick,red,|-|] (0.25,0)--(2.25,0);
    \node at (1.25,0) {$A$};
    \draw[thick,|-|] (2.75,0)--(4.75,0);
    \node at (3.75,0) {$B$};
    \node at (3.75,0.7) {\scriptsize ?};
  \end{scope}
  \node[above left] at (8,-3) {$s_2$};
  \draw (8,-3)--(8,3);
  \begin{scope}[shift={(10.5,1.2)}]
    \draw[thick,red,|-|] (0.25,0)--(2.25,0);
    \node at (1.25,0) {$A$};
    \draw[thick,blue,|-|] (2.75,0)--(4.75,0);
    \node at (3.75,0) {$B$};
    \draw[thick,green,|-|] (1.25,1)--(3.75,1);
    \node at (2.5,1) {$C$};
  \end{scope}
  \node[above left] at (18.2,0) {$s_3$};
  \draw (8,0)--(18.2,0);
  \begin{scope}[shift={(8,-1)}]
    \draw[thick,red,|-|] (0.25,0)--(2.25,0);
    \node at (1.25,0) {$A$};
    \draw[thick,red,|-|] (2.75,0)--(4.75,0);
    \node at (3.75,0) {$B$};
    \draw[thick,blue,|-|] (0.6,-1)--(2.6,-1);
    \node at (1.6,-1) {$D$};
  \end{scope}
  \node[above left] at (13,-3) {$s_4$};
  \draw (13,-3)--(13,0);
  \begin{scope}[shift={(13,-1)}]
    \draw[thick,red,|-|] (0.25,0)--(2.25,0);
    \node at (1.25,0) {$A$};
    \draw[thick,red,|-|] (2.75,0)--(4.75,0);
    \node at (3.75,0) {$B$};
    \draw[thick,blue,|-|] (0.6,-1)--(2.6,-1);
    \node at (1.6,-1) {$D$};
    \draw[thick,green,|-|] (2.4,-1.5)--(4.4,-1.5);
    \node at (3.4,-1.5) {$E$};
  \end{scope}
  \node[above left] at (18.2,-3) {$s_5$};
\end{tikzpicture}\\[3ex]
\begin{tikzpicture}[xscale=0.7,>=latex]
  \tikzstyle{every node}=[inner sep=2pt,fill=white]
  \node (a) at (0,0) {$s_1$};
  \node (b) at (2.5,0) {$s_2$};
  \node (c) at (5,1) {$s_3$};
  \node (d) at (4,-1) {$s_4$};
  \node (e) at (6,-1.5) {$s_5$};
  \path (a) edge[bend left=20, thick] (c);
  \path (b) edge[bend left=30, thick] (c);
  \path (a) edge[bend right=20, thick] (d);
  \path (d) edge[bend right=25, thick] (e);
  \path (b) edge[bend left=20, thick] (e);
  \draw[dotted,thick,->] (a)--(b);
  \draw[dotted,thick,->] (b)--(d);
  \draw[dotted,thick,->] (d)--(e);
  \draw[dotted,thick,->] (b)--(c);
  \begin{scope}[xshift=7.7cm]
  \draw[thick,|-|] (0.25,0)--(2.25,0);
  \node (a) at (1.25,0) {$A$};
  \draw[thick,|-|] (2.75,0)--(4.75,0);
  \node (b) at (3.75,0) {$B$};
  \draw[thick,|-|] (1.25,1)--(3.75,1);
  \node (c) at (2.5,1) {$C$};
  \draw[thick,|-|] (0.6,-1)--(2.6,-1);
  \node (d) at (1.6,-1) {$D$};
  \draw[thick,|-|] (2.4,-1.5)--(4.4,-1.5);
  \node (e) at (3.4,-1.5) {$E$};
  \path (a.40) edge[bend left=23,dotted,thick,->] (b.150);
  \path (b.120) edge[dotted,thick,->] (c);
  \path (b.-150) edge[dotted,thick,->] (d.30);
  \path (d.-30) edge[dotted,thick,->] (e.150);
  \end{scope}
  \node at (15,-0.25) {$\begin{aligned}
    \mu(s_1)&=A\\
    \mu(s_2)&=B\\
    \mu(s_3)&=C\\
    \mu(s_4)&=D\\
    \mu(s_5)&=E
  \end{aligned}$};
\end{tikzpicture}
\caption[]{A strategy of Presenter forcing $3$ colors in the game $\gamINT(2)$.
In the first two rounds, Presenter introduces two disjoint intervals $A$ and $B$.
If they receive different colors, then Presenter forces a third color in the next round by presenting $C$.
If $A$ and $B$ receive the same color, then Presenter forces a third color in the next two rounds by presenting $D$ and $E$.
The five presentation scenarios $s_1,\ldots,s_5$ that can occur form a game graph of $\gamINT(2)$, illustrated as an abstract graph (on the left) and with representation by the intervals $A,B,C,D,E$ (on the right).}
\label{fig:int-game}
\end{figure}

Now, we derive Proposition \ref{prop:rectangle} from a known result on the game $\gamINT(k)$.

\begin{theorem}[Kierstead, Trotter \cite{KT81}]
\label{thm:int-game}
The value of the game\/ $\gamINT(k)$ is\/ $3k-2$.
In particular, Presenter has a finite strategy to force Algorithm to use at least\/ $3k-2$ colors in\/ $\gamINT(k)$.
\end{theorem}

\begin{proof}[Proof of Proposition \ref{prop:rectangle}]
By Theorem \ref{thm:int-game} and Lemma \ref{lem:game-lower}, there are game graphs of $\gamINT(k)$ with chromatic number $3k-2$ (see Figure \ref{fig:int-game}).
It remains to show that every game graph of $\gamINT(k)$ has an intersection model by axis-parallel rectangles.

Let $G$ be a game graph of $\gamINT(k)$ with underlying forest $F$ on $V(G)$ and with representation $\mu\colon V(G)\to\calI$.
For $u\in V(G)$, let $F(u)$ denote the set of vertices of the subtree of $F$ rooted at $u$, inclusive.
We run depth-first search on $F$ and record, for each $u\in V(G)$, the times $x_u,y_u\in\setZ$ at which $F(u)$ is entered and left, respectively, so that
\begin{itemize}
\item $x_u<y_u$ for every $u\in V(G)$,
\item if $v\in F(u)\setminus\{u\}$, then $x_u<x_v<y_v<y_u$,
\item if $v\notin F(u)$ and $u\notin F(v)$, then $[x_u,y_u]\cap[x_v,y_v]=\emptyset$.
\end{itemize}
For every vertex $u\in V(G)$, let $R_u$ be the rectangle in $\setR^2$ defined by $R_u=\mu(u)\times[x_u,y_u]$ (see Figure \ref{fig:game-to-rectangle}).
Consider any two vertices $u,v\in V(G)$.
If $v\in F(u)$ or $u\in F(v)$, then $[x_v,y_v]\subset[x_u,y_u]$ or $[x_u,y_u]\subset[x_v,y_v]$, respectively; hence, $R_u$ and $R_v$ intersect if and only if $\mu(u)$ and $\mu(v)$ intersect, that is, if and only if $uv\in E(G)$.
If $v\notin F(u)$ and $u\notin F(v)$, so that $uv\notin E(G)$, then $[x_u,y_u]\cap[x_v,y_v]=\emptyset$, and thus $R_u\cap R_v=\emptyset$.
This shows that the mapping $u\mapsto R_u$ is an intersection model of $G$.
\end{proof}

\begin{figure}[t]
\begin{tikzpicture}[xscale=0.9,>=latex,shorten >=-0.4pt,shorten <=-0.4pt]
  \tikzstyle{every node}=[inner sep=2pt,fill=white]
  \draw[thick,|-|] (0.75,-1)--(0.75,5);
  \node (a) at (0.75,2) {$a$};
  \node (a) at (0.75,2) {$a$};
  \draw[thick,|-|] (2.5,-0.5)--(2.5,4.5);
  \node (b) at (2.5,2) {$b$};
  \draw[thick,|-|] (5,2.5)--(5,4);
  \node (c) at (5,3.25) {$c$};
  \draw[thick,|-|] (4,0)--(4,1.5);
  \node (d) at (4,0.75) {$d$};
  \draw[thick,|-|] (6,0.4)--(6,1.1);
  \node (e) at (6,0.75) {$e$};
  \path (a) edge[thick,bend left=22] (c);
  \path (b) edge[thick,bend left=30] (c);
  \path (a) edge[thick,bend right=20] (d);
  \path (d) edge[thick,bend right=30] (e);
  \path (b) edge[thick,bend left=16] (e);
  \draw[dotted,thick,->] (a)--(b);
  \draw[dotted,thick,->] (b)--(d);
  \draw[dotted,thick,->] (d)--(e);
  \draw[dotted,thick,->] (b)--(c);
  \begin{scope}
    \tikzstyle{every node}=[inner sep=5pt]
    \node[left] at (0.75,-1) {$x_a$};
    \node[left] at (0.75,5) {$y_a$};
    \node[left] at (2.5,-0.5) {$x_b$};
    \node[left] at (2.5,4.5) {$y_b$};
    \node[left] at (5,2.5) {$x_c$};
    \node[left] at (5,4) {$y_c$};
    \node[left] at (4,0) {$x_d$};
    \node[left] at (4,1.5) {$y_d$};
    \node[left,inner sep=1pt] at (6,0.25) {$x_e$};
    \node[left,inner sep=2pt] at (6,1.25) {$y_e$};
  \end{scope}
  \draw[thick,|-|] (8,7)--(11,7);
  \node (a3) at (9.5,7) {$\mu(a)$};
  \draw[thick,|-|] (12,7)--(15,7);
  \node (b3) at (13.5,7) {$\mu(b)$};
  \draw[thick,|-|] (8.5,7.5)--(14.5,7.5);
  \node (c3) at (11.5,7.5) {$\mu(c)$};
  \draw[thick,|-|] (9,6.5)--(11.7,6.5);
  \node (d3) at (10.35,6.5) {$\mu(d)$};
  \draw[thick,|-|] (11.3,6)--(14,6);
  \node (e3) at (12.65,6) {$\mu(e)$};
  \draw[fill=gray,opacity=0.5] (8,5) rectangle (11,-1);
  \draw[fill=gray,opacity=0.5] (12,4.5) rectangle (15,-0.5);
  \draw[fill=gray,opacity=0.5] (8.5,4) rectangle (14.5,2.5);
  \draw[fill=gray,opacity=0.5] (9,1.5) rectangle (11.7,0);
  \draw[fill=gray,opacity=0.5] (11.3,1.1) rectangle (14,0.4);
  \draw[dashed] (0.75,5)--(11,5);
  \draw[dashed] (0.75,-1)--(11,-1);
  \draw[dashed] (8,7)--(8,-1);
  \draw[dashed] (11,7)--(11,-1);
  \draw[dashed] (2.5,4.5)--(15,4.5);
  \draw[dashed] (2.5,-0.5)--(15,-0.5);
  \draw[dashed] (12,7)--(12,-0.5);
  \draw[dashed] (15,7)--(15,-0.5);
  \draw[dashed] (5,4)--(14.5,4);
  \draw[dashed] (5,2.5)-- (14.5,2.5);
  \draw[dashed] (8.5,7.5)--(8.5,2.5);
  \draw[dashed] (14.5,7.5)--(14.5,2.5);
  \draw[dashed] (4,1.5)--(11.7,1.5);
  \draw[dashed] (4,0)--(11.7,0);
  \draw[dashed] (9,6.5)--(9,0);
  \draw[dashed] (11.7,6.5)--(11.7,0);
  \draw[dashed] (6,1.1)--(14,1.1);
  \draw[dashed] (6,0.4)--(14,0.4);
  \draw[dashed] (11.3,6.0)--(11.3,0.4);
  \draw[dashed] (14,6.0)--(14,0.4);
  \draw[thick] (8,5) rectangle (11,-1);
  \draw[thick] (12,4.5) rectangle (15,-0.5);
  \draw[thick] (8.5,4) rectangle (14.5,2.5);
  \draw[thick] (9,1.5) rectangle (11.7,0);
  \draw[thick] (11.3,1.1) rectangle (14,0.4);
\end{tikzpicture}
\caption[]{Representation of a game graph of $\gamINT(2)$ as an intersection graph of axis-parallel rectangles.}
\label{fig:game-to-rectangle}
\end{figure}
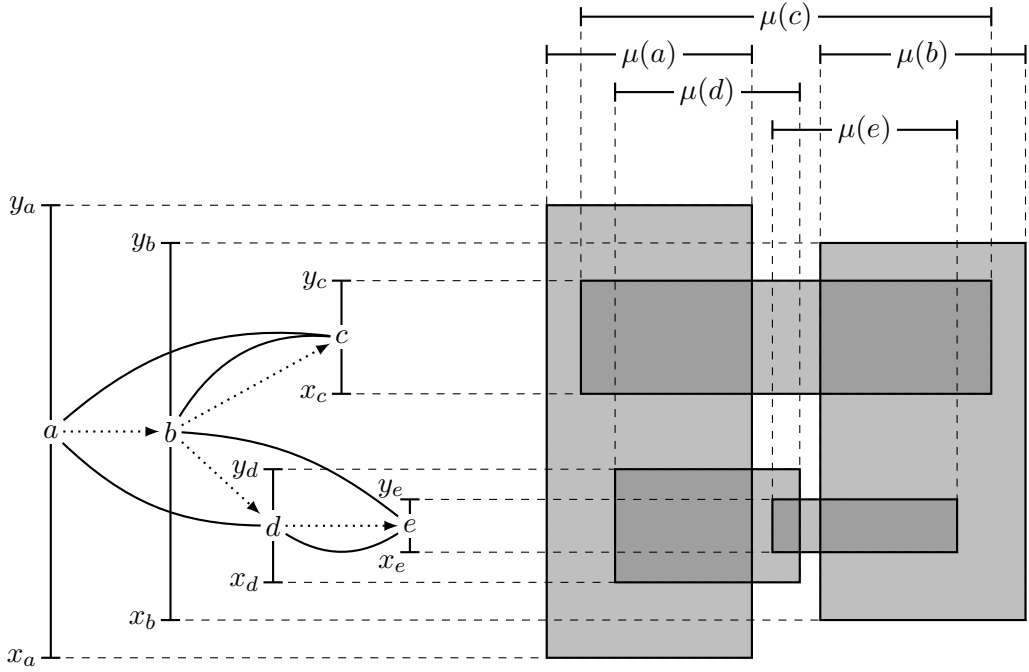

In an analogous way, we can reprove the result of Burling \cite{Bur65} that there exist triangle-free intersection graphs of axis-parallel boxes in $\setR^3$ with chromatic number $\Theta(\log\log n)$.
To this end, we use the result of Erlebach and Fiala \cite{EF02} that Presenter can force the use of arbitrarily many colors in the on-line coloring game on the class of triangle-free rectangle graphs presented with their representation by axis-parallel rectangles.
Their strategy (a geometric realization of the strategy for forests mentioned in the introduction) forces the use of $c$ colors in $2^{c-1}$ rounds with $2^{2^{\smash{O(c)}}}$ presentation scenarios.
Hence, Lemma \ref{lem:game-lower} gives us a triangle-free game graph with chromatic number $c$ and with $2^{2^{\smash{O(c)}}}$ vertices.
The same argument as in the proof of Proposition \ref{prop:rectangle}, using an additional dimension to encode the branching structure of the game graph, shows that this graph is an intersection graph of axis-parallel boxes in $\setR^3$.
The graphs obtained this way are the same as the graphs constructed by Burling and isomorphic to the triangle-free rectangle overlap graphs with chromatic number $\Theta(\log\log n)$ constructed in \cite{PKK+13}.

\section{Interval filament graphs}
\label{sec:filament}

This section is devoted to the proof of Theorem \ref{thm:filament}.
Let $\dom(f)$ denote the domain of an interval filament $f$, that is, the closed interval on which the function $f$ is defined.
We will assume, without loss of generality, that in any interval filament intersection model, the domains are in general position, that is, no two of their endpoints coincide.

The following lemma allows us to reduce the general problem of coloring interval filament graphs to the problem for domain-non-overlapping interval filament graphs.

\begin{lemma}
\label{lem:reduction-to-non-overlapping}
Let\/ $g\colon\setN\to\setN$ be a non-decreasing function with the property that every interval overlap graph satisfies\/ $\chi\leq g(\omega)$.
Then the vertices of every interval filament graph can be partitioned into at most\/ $g(\omega)$ classes so that the subgraph induced on each class is a domain-non-overlapping interval filament graph.
\end{lemma}

\begin{proof}
Let $G$ be a graph with an interval filament intersection model $u\mapsto f_u$.
Let $G'$ be the graph with $V(G')=V(G)$ such that $uv\in E(G')$ if and only if $\dom(f_u)$ and $\dom(f_v)$ overlap.
It follows that $G'$ is a subgraph of $G$ and $G'$ is an interval overlap graph with overlap model $u\mapsto\dom(f_u)$.
The definition of $g$ implies that $G'$ can be properly colored using at most $g(\omega(G'))$ colors, which is at most $g(\omega(G))$ colors due to the monotonicity of $g$.
The model $u\mapsto f_u$ restricted to each color class consists of interval filaments with non-overlapping domains.
\end{proof}

The \emph{incomparability graph} of a partial order $<$ on a set $P$ is the graph with vertex set $P$ and edge set consisting of the pairs of $<$-incomparable elements of $P$.
A graph $G$ is a \emph{co-comparability graph} if it is the incomparability graph of some partial order on $V(G)$.
Consider an on-line game $\gamCOCO(k)$ on the class of co-comparability graphs with clique number at most $k$ presented with their order relation in the \emph{up-growing} manner.
That is, Presenter builds a co-comparability graph $G$ declaring, in each round, the order relation $<$ between the vertices presented before and the new vertex so that
\begin{enumerate}
\item\label{item:coco-graph} $G$ is the incomparability graph of the order $<$ on $V(G)$,
\item\label{item:coco-prec} every vertex of $G$ is maximal in the order $<$ at the moment it is presented,
\item\label{item:coco-omega} $\omega(G)\leq k$, that is, the width of the order $<$ is at most $k$,
\end{enumerate}
and Algorithm properly colors $G$ on-line.

\begin{lemma}
\label{lem:coco-graph}
A graph\/ $G$ is a game graph of\/ $\gamCOCO(k)$ if and only if\/ $G$ is a domain-non-overlapping interval filament graph and\/ $\omega(G)\leq k$.
\end{lemma}

\begin{proof}
Let $G$ be a graph with a domain-non-overlapping interval filament intersection model $u\mapsto f_u$ and with $\omega(G)\leq k$.
The inclusion order on the domains of the interval filaments $f_u$ defines a forest $F$ on $V(G)$ so that for each $v\in V(G)$,
\begin{itemize}
\item if there is no $u\in V(G)$ such that $\dom(f_u)\supset\dom(f_v)$, then $v$ is a root of $F$,
\item otherwise, the parent of $v$ in $F$ is the unique $u\in V(G)$ such that $\dom(f_u)\supset\dom(f_v)$ and $\dom(f_u)$ is minimal with this property.
\end{itemize}
It follows that $u$ is an ancestor of $v$ in $F$ if and only if $\dom(f_u)\supset\dom(f_v)$.
We define a relation $<$ on $V(G)$ so that $u<v$ if and only if $\dom(f_u)\supset\dom(f_v)$ and $f_u\cap f_v=\emptyset$.
Clearly, $<$ is a partial order.
Consider the path $P_v$ in $F$ from a root to a vertex $v$.
The graph $G[V(P_v)]$, the order $<$ restricted to $V(P_v)$, and the order $\prec$ of vertices along $P_v$ form a valid $|V(P_v)|$-round presentation scenario in $\gamCOCO(k)$.
Indeed, the condition \ref{item:coco-graph} of $\gamCOCO(k)$ holds, because if $u\prec v$, then $\dom(f_u)\supset\dom(f_v)$, so $u<v$ if and only if $uv\notin E(G)$; \ref{item:coco-prec} holds, because if $u<v$, then $\dom(f_u)\supset\dom(f_v)$, so $u\prec v$; and \ref{item:coco-omega} follows from the assumption that $\omega(G)\leq k$.
Moreover, if $uv\in E(G)$, then $f_u\cap f_v\neq\emptyset$, which implies $\dom(f_u)\subset\dom(f_v)$ or $\dom(f_u)\supset\dom(f_v)$, by the assumption that the model $u\mapsto f_u$ is domain-non-overlapping.
Hence, if $uv\in E(G)$, then $u$ is an ancestor of $v$ or $v$ is an ancestor of $u$.
This shows that $G$ is indeed a game graph of $\gamCOCO(k)$.

For the converse implication, we use a result due to Golumbic, Rotem and Urrutia \cite{GRU83} and Lovász \cite{Lov83}, which asserts that every partial order is isomorphic to the order $<$ on some family of continuous functions $[0,1]\to(0,\infty)$, where $f<g$ means that $f(x)<g(x)$ for every $x\in[0,1]$.
Let $G$ be a game graph of $\gamCOCO(k)$ with underlying forest $F$ and relation $<$.
For $u\in V(G)$, let $F(u)$ denote the set of vertices of the subtree of $F$ rooted at $u$, including $u$ itself.
As in the proof of Proposition \ref{prop:rectangle}, we use depth-first search to compute, for each $u\in V(G)$, numbers $x_u,y_u\in\setZ$ such that
\begin{itemize}
\item $x_u<y_u$ for every $u\in V(G)$,
\item if $v\in F(u)\setminus\{u\}$, then $x_u<x_v<y_v<y_u$,
\item if $v\notin F(u)$ and $u\notin F(v)$, then $[x_u,y_u]\cap[x_v,y_v]=\emptyset$.
\end{itemize}

Let $L$ denote the set of leaves of $F$, and let $L(u)=L\cap F(u)$ for $u\in V(G)$.
For $v\in L$, let $P_v$ denote the path in $F$ from a root to $v$.
The graph $G[V(P_v)]$ is the incomparability graph of the order $<$ restricted to $V(P_v)$.
Hence, by the above-mentioned result of \cite{GRU83,Lov83}, it has an intersection representation by continuous functions $[x_v,y_v]\to(0,\infty)$.
Specifically, every vertex $u\in V(P_v)$ can be assigned a continuous function $f_{u,v}\colon[x_v,y_v]\to(0,\infty)$ so that $u_1<u_2$ if and only if $f_{u_1,v}>f_{u_2,v}$ for any $u_1,u_2\in V(P_v)$ (note that the order is reversed).
Now, for every vertex $u\in V(G)$, we define an interval filament $f_u$ as the union of the following curves:
\begin{itemize}
\item the functions $f_{u,v}$ for all $v\in L(u)$,
\item the segment connecting points $(x_u-\frac{1}{3},0)$ and $(x_v,f_{u,v}(x_v))$ for the first leaf $v\in L(u)$ in the depth-first search order,
\item the segments connecting points $(y_{v_1},f_{u,v_1}(y_{v_1}))$ and $(x_{v_2},f_{u,v_2}(x_{v_2}))$ for any two leaves $v_1,v_2\in L(u)$ consecutive in the depth-first search order,
\item the segment connecting points $(y_v,f_{u,v}(y_v))$ and $(y_u+\frac{1}{3},0)$ for the last leaf $v\in L(u)$ in the depth-first search order.
\end{itemize}
See Figure \ref{fig:filament} for an illustration.
It follows that $\dom(f_u)=[x_u-\frac{1}{3},y_u+\frac{1}{3}]$ for every $u\in V(G)$, so the domains of the interval filaments $f_u$ do not overlap.

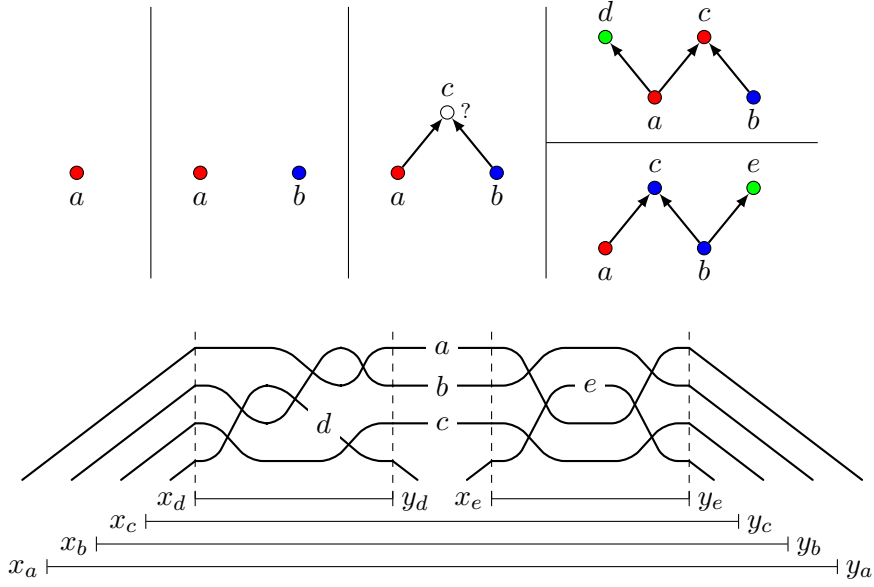
\begin{figure}[t]
\centering
\begin{tikzpicture}[xscale=0.65,yscale=0.4,>=latex]
  \tikzstyle{every node}=[circle,minimum size=5pt,inner sep=0pt,draw,fill]
  \tikzstyle{every label}=[rectangle,text height=1.5ex,text depth=0.25ex,inner sep=2pt,draw=none,fill=none]
  \begin{scope}[shift={(0.5,0)}]
    \node[label=below:$a$,fill=red] (a1) at (1,-1) {};
  \end{scope}
  \draw (3,-4.5)--(3,4.5);
  \begin{scope}[shift={(3,0)}]
    \node[label=below:$a$,fill=red] (a1) at (1,-1) {};
    \node[label=below:$b$,fill=blue] (b1) at (3,-1) {};
  \end{scope}
  \draw (7,-4.5)--(7,4.5);
  \begin{scope}[shift={(7,0)}]
    \node[label=below:$a$,fill=red] (a1) at (1,-1) {};
    \node[label=below:$b$,fill=blue] (b1) at (3,-1) {};
    \node[label=above:$c$,label={[yshift=1pt]right:\scriptsize ?},fill=white] (c1) at (2,1) {};
    \draw[thick,->] (a1)--(c1);
    \draw[thick,->] (b1)--(c1);
  \end{scope}
  \draw (11,-4.5)--(11,4.5);
  \begin{scope}[shift={(11.2,-2.5)}]
    \node[label=below:$a$,fill=red] (a1) at (1,-1) {};
    \node[label=below:$b$,fill=blue] (b1) at (3,-1) {};
    \node[label=above:$c$,fill=blue] (c1) at (2,1) {};
    \node[label=above:$e$,fill=green] (e1) at (4,1) {};
    \draw[thick,->] (a1)--(c1);
    \draw[thick,->] (b1)--(c1);
    \draw[thick,->] (b1)--(e1);
  \end{scope}
  \draw (11,0)--(16.5,0);
  \begin{scope}[shift={(12.2,2.5)}]
    \node[label=below:$a$,fill=red] (a1) at (1,-1) {};
    \node[label=below:$b$,fill=blue] (b1) at (3,-1) {};
    \node[label=above:$c$,fill=red] (c1) at (2,1) {};
    \node[label=above:$d$,fill=green] (d1) at (0,1) {};
    \draw[thick,->] (a1)--(c1);
    \draw[thick,->] (b1)--(c1);
    \draw[thick,->] (a1)--(d1);
  \end{scope}
\end{tikzpicture}\\[3ex]
\begin{tikzpicture}[xscale=0.65,yscale=0.5]
  \begin{scope}[shorten <=-0.2pt,shorten >=-0.2pt]
    \draw[|-|] (1,-2.8)--(17,-2.8);
    \draw[|-|] (2,-2.2)--(16,-2.2);
    \draw[|-|] (3,-1.6)--(15,-1.6);
    \draw[|-|] (4,-1)--(8,-1);
    \draw[|-|] (10,-1)--(14,-1);
  \end{scope}
  \tikzstyle{every node}=[text height=1ex,text depth=0ex,inner sep=3pt]
  \node[left] at (1,-2.8) {$x_a$};
  \node[right] at (17,-2.8) {$y_a$};
  \node[left] at (2,-2.2) {$x_b$};
  \node[right] at (16,-2.2) {$y_b$};
  \node[left] at (3,-1.6) {$x_c$};
  \node[right] at (15,-1.6) {$y_c$};
  \node[left] at (4,-1) {$x_d$};
  \node[right] at (8,-1) {$y_d$};
  \node[left] at (10,-1) {$x_e$};
  \node[right] at (14,-1) {$y_e$};
  \draw[dashed] (4,-1)--(4,3.5);
  \draw[dashed] (8,-1)--(8,3.5);
  \draw[dashed] (10,-1)--(10,3.5);
  \draw[dashed] (14,-1)--(14,3.5);
  \tikzstyle{every node}=[inner sep=3pt,fill=white]
  \begin{scope}[thick,line cap=round]
    \draw (0.5,-0.5)--(4,3);
    \draw[rounded corners=5pt] (4,3)--(5.8,3)--(6.7,2)--(7.2,2)--(7.6,3)--(8,3);
    \draw (8,3)--(10,3);
    \draw[rounded corners=5pt] (10,3)--(10.5,3)--(11.3,1)--(12.7,1)--(13.5,3)--(14,3);
    \draw (14,3)--(17.5,-0.5);
  \end{scope}
  \node at (9,3) {$a$};
  \begin{scope}[thick,line cap=round]
    \draw (1.5,-0.5)--(4,2);
    \draw[rounded corners=5pt] (4,2)--(4.5,2)--(5.2,1)--(5.7,1)--(6.7,3)--(7.2,3)--(7.6,2)--(8,2);
    \draw (8,2)--(10,2);
    \draw[rounded corners=5pt] (10,2)--(10.5,2)--(11.2,3)--(12.8,3)--(13.5,2)--(14,2);
    \draw (14,2)--(16.5,-0.5);
  \end{scope}
  \node at (9,2) {$b$};
  \begin{scope}[thick,line cap=round]
    \draw (2.5,-0.5)--(4,1);
    \draw[rounded corners=5pt] (4,1)--(4.4,1)--(5.1,0)--(6.8,0)--(7.5,1)--(8,1);
    \draw (8,1)--(10,1);
    \draw[rounded corners=5pt] (10,1)--(10.5,1)--(11.2,0)--(12.8,0)--(13.5,1)--(14,1);
    \draw (14,1)--(15.5,-0.5);
  \end{scope}
  \node at (9,1) {$c$};
  \begin{scope}[thick,line cap=round]
    \draw (3.5,-0.5)--(4,0);
    \draw[rounded corners=5pt] (4,0)--(4.4,0)--(5.2,2)--(5.7,2)--(7.5,0)--(8,0);
    \draw (8,0)--(8.5,-0.5);
  \end{scope}
  \node at (6.6,1) {$d$};
  \begin{scope}[thick,line cap=round]
    \draw (9.5,-0.5)--(10,0);
    \draw[rounded corners=5pt] (10,0)--(10.5,0)--(11.3,2)--(12.7,2)--(13.5,0)--(14,0);
    \draw (14,0)--(14.5,-0.5);
  \end{scope}
  \node at (12,2) {$e$};
\end{tikzpicture}
\caption[]{Top: A strategy of Presenter forcing $3$ colors in $4$ rounds of the game $\gamCOCO(2)$.
If $a$, $b$, $c$ receive distinct colors, then Presenter wins in $3$ rounds.
Otherwise, the color of $c$ is the same as the color of $a$ or $b$, and depending on Algorithm's choice, Presenter forces a $3$rd color in the $4$th round.
Bottom: A domain-non-overlapping interval filament model of the game graph arising from the strategy on the top.}
\label{fig:filament}
\end{figure}

It remains to prove that $u\mapsto f_u$ is an intersection model of $G$.
Fix $u,v\in V(G)$.
First, suppose $v\in F(u)$ and $u<v$, so that $uv\notin E(G)$.
By the definition of $f_u$ and $f_v$, we have $\dom(f_u)\supset\dom(f_v)$, and $f_u$ lies entirely above $f_v$.
Hence $f_u\cap f_v=\emptyset$.
Now, suppose $v\in F(u)$ and $u\not<v$.
We also have $v\not<u$, by the condition \ref{item:coco-prec} of the definition of $\gamCOCO(k)$.
Hence $uv\in E(G)$.
For any leaf $w\in L(v)$, the functions $f_{u,w}$ and $f_{v,w}$ intersect, so $f_u\cap f_v\neq\emptyset$.
The case that $u\in F(v)$ is analogous.
Finally, suppose $u\notin F(v)$ and $v\notin F(u)$, so that $uv\notin E(G)$.
It follows that $[x_u,y_u]\cap[x_v,y_v]=\emptyset$, so $\dom(f_u)\cap\dom(f_v)=\emptyset$.
Hence $f_u\cap f_v=\emptyset$.
This shows that $u\mapsto f_u$ is indeed an intersection model of $G$.
\end{proof}

\begin{theorem}[Felsner \cite{Fel97}]
\label{thm:coco-game}
The value of the game\/ $\gamCOCO(k)$ is\/ $\smash[b]{\binom{k+1}{2}}$.
That is, there is an on-line algorithm using at most\/ $\smash[b]{\binom{k+1}{2}}$ colors, and there is a finite strategy of Presenter forcing Algorithm to use at least\/ $\binom{k+1}{2}$ colors in\/ $\gamCOCO(k)$.
\end{theorem}

Theorem \ref{thm:filament} \ref{item:non-overlapping-filament-upper}--\ref{item:non-overlapping-filament-lower} follows from Theorem \ref{thm:coco-game}, Lemma \ref{lem:coco-graph}, and Lemmas \ref{lem:game-upper} and \ref{lem:game-lower} (respectively).
Theorem \ref{thm:filament}~\ref{item:filament-upper} follows from Theorem \ref{thm:filament}~\ref{item:non-overlapping-filament-upper} and Lemma \ref{lem:reduction-to-non-overlapping}.

\section{Reduction to clean overlap graphs}
\label{sec:reduction-to-clean}

Recall that an overlap graph is \emph{clean} if it has an overlap model such that no two overlapping sets both contain a third set.
The goal of this section is to establish the following reduction of the general problem of coloring overlap graphs to the problem for clean overlap graphs.

\begin{theorem}
\label{thm:reduction-to-clean}
Let\/ $G$ be an overlap graph.
If every clean induced subgraph\/ $H$ of\/ $G$ with\/ $\omega(H)\leq j$ satisfies\/ $\chi(H)\leq\alpha_j$ for\/ $2\leq j\leq\omega(G)$, then\/ $\chi(G)\leq 2^{\omega(G)-1}\alpha_2\cdots\alpha_{\omega(G)}$.
\end{theorem}

It is proved in \cite{KPW15} that every triangle-free overlap graph can be partitioned into two clean graphs: the union of odd levels and the union of even levels in the breadth-first search forest.
This proves Theorem \ref{thm:reduction-to-clean} for graphs with clique number $2$.
However, such a simple partition is insufficient for graphs with clique number greater than $2$.
We will need the following generalization of breadth-first search, which we call \emph{$k$-clique breadth-first search}.

\begin{center}
\begin{algorithm}[H]
\BlankLine
\textbf{$k$-clique breadth-first search}\par
\SetKwInOut{Input}{input}
\SetKwInOut{Output}{output}
\Input{a graph $G$ with vertices ordered as $v_1,\ldots,v_n$}
\Output{a partition of $\{v_1,\ldots,v_n\}$ into sets $L_d$ with $d\geq 0$}
\BlankLine
$V\coloneq\{v_1,\ldots,v_n\}$;\quad $d\coloneq 0$\;
\While{$V\neq\emptyset$}{
  \uIf{\upshape there is a $k$-clique $K$ with $|K\cap V|=1$}{
    $L_d\coloneq\{v_j\in V\colon$there is a $k$-clique $K$ with $K\cap V=\{v_j\}\}$\;
  }
  \Else{
    pick $v_i\in V$ with minimum index $i$\;
    $L_d\coloneq\{v_i\}$\;
  }
  $V\coloneq V\setminus L_d$;\quad $d\coloneq d+1$\;
}
\end{algorithm}
\end{center}

\allowbreak
See Figure \ref{fig:bfs} for an illustration of the algorithm.
It is clear that it terminates after time polynomial in $n$ (for fixed $k$).
The $2$-clique breadth-first search is just the ordinary breadth-first search: every connected component of $G$ is the union of some consecutive sets $L_d,\ldots,L_{d+t}$, of which $L_{d+i}$ is the set of vertices at distance $i$ from the vertex with minimum index in that connected component.
The following two properties of the $k$-clique breadth-first search generalize those of the ordinary breadth-first search.

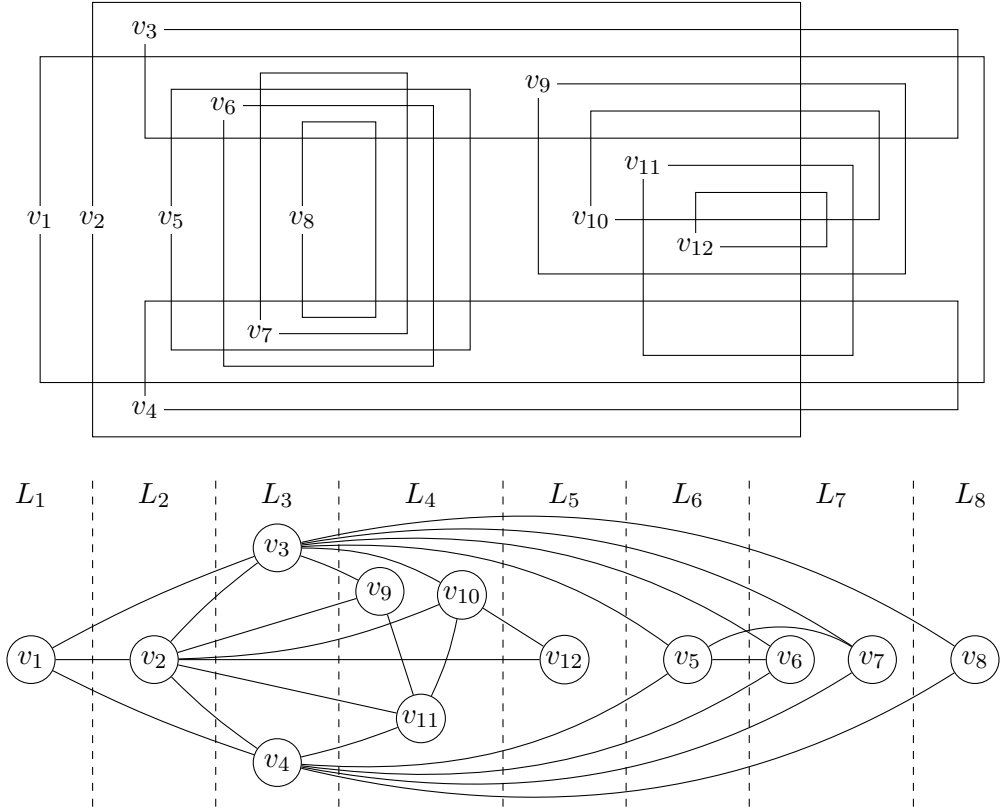
\begin{figure}[t]
\centering
\begin{tikzpicture}[xscale=1.38,yscale=0.72]
  \tikzstyle{every node}=[inner sep=2pt,fill=white]
  \draw (0,1) rectangle (9,7);
  \node at (0,4) {$v_1$};
  \draw (0.5,0) rectangle (7.25,8);
  \node at (0.5,4) {$v_2$};
  \draw (1,5.5) rectangle (8.75,7.5);
  \node at (1,7.5) {$v_3$};
  \draw (1,0.5) rectangle (8.75,2.5);
  \node at (1,0.5) {$v_4$};
  \draw (1.25,1.6) rectangle (4.1,6.4);
  \node at (1.25,4) {$v_5$};
  \draw (1.75,1.3) rectangle (3.75,6.1);
  \node at (1.75,6.1) {$v_6$};
  \draw (2.1,1.9) rectangle (3.5,6.7);
  \node at (2.1,1.9) {$v_7$};
  \draw (2.5,2.2) rectangle (3.2,5.8);
  \node at (2.5,4) {$v_8$};
  \draw (4.75,3) rectangle (8.25,6.5);
  \node at (4.75,6.5) {$v_9$};
  \draw (5.25,4) rectangle (8,6.0);
  \node at (5.25,4.0) {$v_{10}$};
  \draw (5.75,1.5) rectangle (7.75,5.0);
  \node at (5.75,5) {$v_{11}$};
  \draw (6.25,3.5) rectangle (7.5,4.5);
  \node at (6.25,3.5) {$v_{12}$};
\end{tikzpicture}\\[3ex]
\begin{tikzpicture}[xscale=0.54,yscale=0.39]
  \tikzstyle{every node}=[circle,minimum size=18pt,inner sep=1pt,draw];
  \node (v1) at (0.5,0) {$v_1$};
  \draw[dashed] (2,-5)--(2,6.1);
  \node (v2) at (3.5,0) {$v_2$};
  \draw[dashed] (5,-5)--(5,6.1);
  \node (v3) at (6.5,3.8) {$v_3$};
  \node (v4) at (6.5,-3.5) {$v_4$};
  \draw[dashed] (8,-5)--(8,6.1);
  \node (v9) at (9,2.32) {$v_9$};
  \node (v10) at (11,2.2) {$v_{10}$};
  \node (v11) at (10,-2) {$v_{11}$};
  \draw[dashed] (12,-5)--(12,6.1);
  \node (v12) at (13.5,0) {$v_{12}$};
  \draw[dashed] (15,-5)--(15,6.1);
  \node (v5) at (16.5,0) {$v_5$};
  \draw[dashed] (18,-5)--(18,6.1);
  \node (v6) at (19,0) {$v_6$};
  \node (v7) at (21,0) {$v_7$};
  \draw[dashed] (22,-5)--(22,6.1);
  \node (v8) at (23.5,0) {$v_8$};
  \path (v1) edge (v2);
  \path (v2) edge[bend left=5] (v3);
  \path (v2) edge[bend right=5] (v4);
  \path (v1) edge[bend left=5] (v3);
  \path (v1) edge[bend right=5] (v4);
  \path (v2) edge (v9);
  \path (v2) edge[bend right=14] (v10);
  \path (v2) edge (v11);
  \path (v2) edge (v12);
  \path (v3) edge[bend left=5] (v9);
  \path (v3) edge[bend left=19] (v10);
  \path (v9) edge (v11);
  \path (v10) edge[bend left=5] (v11);
  \path (v4) edge[bend right=5] (v11);
  \path (v10) edge (v12);
  \path (v5) edge (v6);
  \path (v5) edge[bend left=41] (v7);
  \path (v3) edge[bend left=25] (v5);
  \path (v3) edge[bend left=26] (v6);
  \path (v3) edge[bend left=28] (v7);
  \path (v3) edge[bend left=30] (v8);
  \path (v4) edge[bend right=25] (v5);
  \path (v4) edge[bend right=26] (v6);
  \path (v4) edge[bend right=28] (v7);
  \path (v4) edge[bend right=30] (v8);
  \tikzstyle{every node}=[inner sep=2pt]
  \node at (0.5,5.6) {$L_1$};
  \node at (3.5,5.6) {$L_2$};
  \node at (6.5,5.6) {$L_3$};
  \node at (10,5.6) {$L_4$};
  \node at (13.5,5.6) {$L_5$};
  \node at (16.5,5.6) {$L_6$};
  \node at (20,5.6) {$L_7$};
  \node at (23.4,5.6) {$L_8$};
\end{tikzpicture}
\caption[]{An illustration of the $3$-clique breadth-first search applied to the rectangle overlap graph above.
The sets $L_1$, $L_2$, $L_6$ and $L_8$ are determined by the \textbf{else} statement of the main loop.}
\label{fig:bfs}
\end{figure}

\begin{lemma}
\label{lem:bfs1}
Let\/ $L_d$ be the sets computed by the\/ $k$-clique breadth-first search on a graph\/ $G$.
Then every\/ $k$-clique in\/ $G$ has two of its vertices in one set\/ $L_d$ or in two consecutive sets\/ $L_d$ and\/ $L_{d+1}$.
\end{lemma}

\begin{proof}
Let $K$ be a $k$-clique in $G$.
Let $L_d$ be the set such that $|K\cap V|\geq 2$ before and $|K\cap V|\leq 1$ after the algorithm performs the assignment $V\coloneq V\setminus L_d$.
It follows that $K\cap L_d\neq\emptyset$.
If $|K\cap L_d|\geq 2$, then $L_d$ satisfies the conclusion of the lemma.
If $|K\cap L_d|=1$, then $|K\cap V|=1$ after the assignment $V\coloneq V\setminus L_d$, so the vertex remaining in $K\cap V$ will be taken to $L_{d+1}$ in the next iteration of the algorithm, which yields $K\cap L_{d+1}\neq\emptyset$.
\end{proof}

\begin{lemma}
\label{lem:bfs2}
Let\/ $G$ be an overlap graph with overlap model\/ $\mu$, with vertices\/ $v_1,\ldots,v_n$ ordered so that $\mu(v_i)\not\subset\mu(v_j)$ for\/ $i<j$, and with\/ $\omega(G)\leq k$.
Then every set\/ $L_d$ computed by the\/ $k$-clique breadth-first search on\/ $G$ induces a clean overlap subgraph of\/ $G$.
\end{lemma}

\begin{proof}
First, we prove the following:
\begin{equation*}
\label{eq:L}
\text{if $v_i\in L_d$, $v_r\in L_{d'}$, and $\mu(v_r)\subset\mu(v_i)$, then $d\leq d'$.}\tag{$*$}
\end{equation*}
Let $v_i\in L_d$, let $d'$ be the minimum index such that $L_{d'}$ contains a vertex $v_r$ with $\mu(v_r)\subset\mu(v_i)$, and suppose to the contrary that $d'<d$.
Consider the set $V$ at the point when the algorithm computes $L_{d'}$.
It follows that $v_i,v_r\in V$.
If there was no $k$-clique $K$ with $|K\cap V|=1$, then the algorithm would not set $L_{d'}$ to $\{v_r\}$, because $v_i$ is a candidate with smaller index.
Hence, there is a $k$-clique $K$ with $|K\cap V|=1$, which implies that there is a $k$-clique $K$ with $K\cap V=\{v_r\}$.
For every $v_s\in K\setminus\{v_r\}$, we have $s<r$, and therefore $\mu(v_s)$ and $\mu(v_i)$ overlap: $\mu(v_s)\subset\mu(v_i)$ would contradict the choice of $d'$, and $\mu(v_s)\supset\mu(v_i)\supset\mu(v_r)$ would contradict the fact that $v_r,v_s\in K$.
Hence, $K'=(K\setminus\{v_r\})\cup\{v_i\}$ is a $k$-clique with $K'\cap V=\{v_i\}$, which yields $v_i\in L_{d'}$.
This contradiction completes the proof of \eqref{eq:L}.

Now, suppose that $G[L_d]$ is not clean.
This means that there are $v_i,v_j,v_r\in L_d$ such that $\mu(v_i)$ overlaps $\mu(v_j)$ and $\mu(v_r)\subset\mu(v_i)\cap\mu(v_j)$.
Consider the set $V$ at the point when the algorithm computes $L_d$.
It follows that there is a $k$-clique $K$ with $K\cap V=\{v_r\}$.
By \eqref{eq:L}, for every $v_s\in K\setminus\{v_r\}$, $\mu(v_s)$ is not contained in either of $\mu(v_i)$ and $\mu(v_j)$ and thus overlaps both of them.
Hence, $(K\setminus\{v_r\})\cup\{v_i,v_j\}$ is a $(k+1)$-clique in $G$, which contradicts the assumption that $\omega(G)\leq k$.
\end{proof}

\begin{proof}[Proof of Theorem \ref{thm:reduction-to-clean}]
\def\odd{\mathrm{odd}}
\def\even{\mathrm{even}}
Let $\mu$ be an overlap model of $G$, and let $k=\omega(G)$.
The proof goes by induction on $k$.
The theorem is trivial for $k=1$, so assume that $k\geq 2$ and the theorem holds for graphs with $\omega\leq k-1$.
Order the vertices of $G$ as $v_1,\ldots,v_n$ so that $\mu(v_i)\not\subset\mu(v_j)$ for $i<j$, and run the $k$-clique breadth-first search to obtain a partition of $\{v_1,\ldots,v_n\}$ into sets $L_d$.
By Lemma \ref{lem:bfs2}, every $L_d$ induces a clean subgraph of $G$, so $\chi(G[L_d])\leq\alpha_k$.
Color each $G[L_d]$ properly with the same set of $\alpha_k$ colors, thus obtaining a partition of the vertices of $G$ into color classes $C_1,\ldots,C_{\alpha_k}$.
Each set of the form $C_i\cap L_d$ is an independent set in $G$.
Let $L_{\odd}$ be the union of all sets $L_d$ with $d$ odd and $L_{\even}$ be the union of all sets $L_d$ with $d$ even.
If there is a $k$-clique in $G[C_i\cap L_{\odd}]$, then, by Lemma \ref{lem:bfs1}, it must contain an edge connecting vertices in one set $L_d$ or two consecutive sets $L_d$ and $L_{d+1}$.
The former is impossible, as $C_i\cap L_d$ is independent, while the latter contradicts the definition of $L_{\odd}$.
Hence, $\omega(G[C_i\cap L_{\odd}])\leq k-1$.
Similarly, $\omega(G[C_i\cap L_{\even}])\leq k-1$.
It follows from the induction hypothesis that $\chi(G[C_i\cap L_{\odd}])\leq 2^{k-2}\alpha_2\cdots\alpha_{k-1}$ and $\chi(G[C_i\cap L_{\even}])\leq 2^{k-2}\alpha_2\cdots\alpha_{k-1}$.
This implies $\chi(G)\leq 2^{k-1}\alpha_2\cdots\alpha_k$, as the $2\alpha_k$ sets $C_i\cap L_{\odd}$ and $C_i\cap L_{\even}$ for $1\leq i\leq\alpha_k$ form a partition of the entire set of vertices of $G$.
\end{proof}

The inductive nature of Theorem \ref{thm:reduction-to-clean} is the main obstacle to generalizing the upper bounds of Theorem \ref{thm:subtree}~\ref{item:clean-subtree-upper} and Theorem \ref{thm:rectangle}~\ref{item:clean-rectangle-upper} from clean to non-clean overlap graphs (keeping the same asymptotic bounds).
Furthermore, if we replace $\chi$ by $\chi_3$ in the proof of Theorem \ref{thm:reduction-to-clean}, then it does no longer work.
This is why we are unable to provide the analogue of Theorem \ref{thm:rectangle-K3} for non-clean rectangle overlap graphs.
We wonder whether a reduction similar to Theorem \ref{thm:reduction-to-clean} but avoiding induction is possible.

\section{Rectangle and subtree overlap graphs}
\label{sec:rectangle-subtree}

In this section, we define two on-line games and relate their game graphs to rectangle and subtree overlap graphs.
These relations will be used for the proofs of Theorems \ref{thm:subtree}--\ref{thm:rectangle-K3} in Sections \ref{sec:coloring} and \ref{sec:construction}.
In view of Theorem \ref{thm:reduction-to-clean}, we can restrict our consideration to clean rectangle and subtree overlap graphs.

First, we introduce the on-line game corresponding to clean rectangle overlap graphs, we define clean interval overlap game graphs, and we describe their relation to clean rectangle overlap graphs that has been established in \cite{KPW15,PKK+13}.
Recall that $\calI$ denotes the set of closed intervals in $\setR$.
Let $\leftside(x)$ and $\rightside(x)$ denote the left and the right endpoints of an interval $x\in\calI$, respectively.
Consider an on-line game $\gamIOV(k)$, in which Presenter builds a clean interval overlap graph $G$ and its representation $\mu\colon V(G)\to\calI$ so that
\begin{enumerate}
\item $\mu$ is an overlap model of $G$, that is, $xy\in E(G)$ if and only if $\mu(x)$ overlaps $\mu(y)$,
\item if $x,y\in V(G)$ and $x$ is presented before $y$, then $\leftside(\mu(x))<\leftside(\mu(y))$,
\item $\mu$ is clean, that is, there are no $x,y,z\in V(G)$ such that $\mu(x)$ and $\mu(y)$ overlap and $\mu(z)\subset\mu(x)\cap\mu(y)$,
\item $\omega(G)\leq k$,
\end{enumerate}
and Algorithm properly colors $G$ on-line.
We will assume, without loss of generality, that in any representation $\mu$ presented in the game, the intervals are in general position, that is, no two of their endpoints coincide.
As a consequence of the definition of a game graph, a graph $G$ is a game graph of $\gamIOV(k)$ if there exist a rooted forest $F$ on $V(G)$ and a mapping $\mu\colon V(G)\to\calI$ such that the following conditions, corresponding to the four above, are satisfied (where $\prec$ denotes the ancestor-descendant relation of $F$):
\begin{enumerate}
\item\label{item:iovgg-def-1} $xy\in E(G)$ if and only if $x\prec y$ or $y\prec x$ and $\mu(x)$ overlaps $\mu(y)$,
\item if $x,y\in V(G)$ and $x\prec y$, then $\leftside(\mu(x))<\leftside(\mu(y))$,
\item\label{item:iovgg-def-clean} there are no $x,y,z\in V(G)$ with $x\prec y\prec z$ such that $\mu(x)$ and $\mu(y)$ overlap and $\mu(z)\subset\mu(x)\cap\mu(y)$,
\item $\omega(G)\leq k$.
\end{enumerate}
A graph is a \emph{clean interval overlap game graph} if it is a game graph of $\gamIOV(k)$ for some $k$.
The conditions \ref{item:iovgg-def-1}--\ref{item:iovgg-def-clean} of the characterization above were used in \cite{KPW15} as the definition of clean interval overlap game graphs (called overlap game graphs therein).

\begin{lemma}[Krawczyk, Pawlik, Walczak \cite{KPW15}]
\label{lem:rectangle-to-game}
Every clean interval overlap game graph is a clean rectangle overlap graph.
The vertices of every clean rectangle overlap graph can be partitioned into\/ $O_\omega(1)$ classes so that the subgraph induced on each class is a clean interval overlap game graph.
\end{lemma}

As it is explained in \cite{KPW15}, the correspondence analogous to Lemma \ref{lem:rectangle-to-game} holds between rectangle overlap graphs and \emph{interval overlap game graphs}, which are defined like clean interval overlap game graphs above but omitting the cleanness condition \ref{item:iovgg-def-clean}.

It is proved in \cite{KPW15} that triangle-free clean interval overlap game graphs (and hence, by Lemma \ref{lem:rectangle-to-game}, triangle-free clean rectangle overlap graphs) satisfy $\chi=O(\log\log n)$.
That proof essentially comes down to an on-line algorithm using $O(\log r)$ colors in $r$ rounds of the game $\gamIOV(2)$, a trick with heavy-light decomposition that we explain later, and the application of Lemma \ref{lem:game-upper}.
We will generalize this to game graphs of $\gamIOV(k)$ and thus to clean rectangle overlap graphs with clique number bounded by any constant.
On the other hand, it is proved in \cite{PKK+13} that Presenter has a strategy to force Algorithm to use $c$ colors in $2^{c-1}$ rounds of the game $\gamIOV(2)$.
This strategy (again a realization of the strategy for forests mentioned in the introduction) has $2^{2^{\smash{O(c)}}}$ presentation scenarios.
Hence, by Lemma \ref{lem:game-lower}, there are triangle-free clean interval overlap game graphs (and thus triangle-free clean rectangle overlap graphs) with chromatic number $\Theta(\log\log n)$.

We also define an on-line game $\gamIOV_3(k)$, a variant of $\gamIOV(k)$ in which Algorithm is required to produce a triangle-free coloring instead of a proper coloring.
The rules for Presenter's moves are the same in $\gamIOV(k)$ and $\gamIOV_3(k)$, and therefore the classes of game graphs of $\gamIOV(k)$ and $\gamIOV_3(k)$ are also the same.

Now, we introduce the on-line game corresponding to clean subtree overlap graphs.
Let $G$ be a clean subtree overlap graph with a clean overlap model $x\mapsto S_x$ by subtrees of a tree $T$.
To avoid confusion with vertices of $G$, we call vertices of $T$ \emph{nodes}.
We make $T$ a rooted tree by choosing an arbitrary node $r$ as the root.
We can assume without loss of generality that $r$ belongs to none of the subtrees $S_x$, as a node with this property can always be added to $T$.
For every $x\in V(G)$, we define $r_x$ to be the unique node of $S_x$ that is closest to $r$ in $T$.
We call the nodes $r_x$ \emph{subtree roots}.
If several subtrees $S_{x_0},\ldots,S_{x_m}$ with $S_{x_i}\not\supset S_{x_j}$ for $i<j$ (without loss of generality) have common subtree root $p=r_{x_0}=\cdots=r_{x_m}$, then the following transformation preserves overlaps and proper inclusions between the subtrees, thus keeping $x\mapsto S_x$ a clean overlap model of $G$:
\begin{itemize}
\item replace the edge $pq$ by a path $pp_1\cdots p_mq$, where $q$ is the other end of the tree edge at $p$ going towards $r$ and $p_1,\ldots,p_m$ are new nodes,
\item add $p_1,\ldots,p_m$ to every subtree containing both $p$ and $q$,
\item for $1\leq i\leq m$, add $p_1,\ldots,p_i$ to $S_{x_i}$ and let $r_{x_i}=p_i$.
\end{itemize}
Applying this transformation repeatedly if necessary, we can assume without loss of generality that all the subtree roots are pairwise distinct.
We construct a rooted forest $F$ on $V(G)$ as follows.
A vertex $x\in V(G)$ is a root of $F$ if the path from $r$ to $r_x$ in $T$ contains no subtree roots other than $r_x$.
Otherwise, the parent of $x$ in $F$ is the vertex $y\in V(G)$ such that $r_y$ is the last subtree root before $r_x$ on the path from $r$ to $r_x$ in $T$.

Consider a path $P$ in $F$ from a root to a leaf.
The overlap graph of the subtrees $S_x$ with $x\in V(P)$ is an interval filament graph \cite{ES07}.
Its interval filament intersection model can be constructed as follows.
The roots of all the subtrees $S_x$ with $x\in V(P)$ lie on a common path $Q=q_1\cdots q_m$ in $T$.
For $1\leq i\leq m$, let $T_i$ denote the connected component of $T$ containing $q_i$ after removing all edges of $Q$.
We represent the nodes of $T$ by points in $\setR^2$, as follows.
Each node $q_i$ is represented by the point $(i,0)$.
Each node $t$ in $T_i$ other than $q_i$ is represented by a point $(x_t,1)$, where $x_t\in(i,i+1)$ and all the $x_t$ are distinct.
Now, we can represent each vertex $x\in V(P)$ such that the intersection of $S_x$ and $Q$ is the subpath $q_i\cdots q_j$ of $Q$ by an interval filament that starts in the interval $(i-1,i)$, ends in the interval $(j,j+1)$, and goes above the points representing the nodes in $S_x$ and no other points representing nodes.
Moreover, we can do this so that the interval filaments representing non-adjacent vertices (with nested or disjoint subtrees) do not intersect.
This yields an interval filament intersection model of $G[V(P)]$.

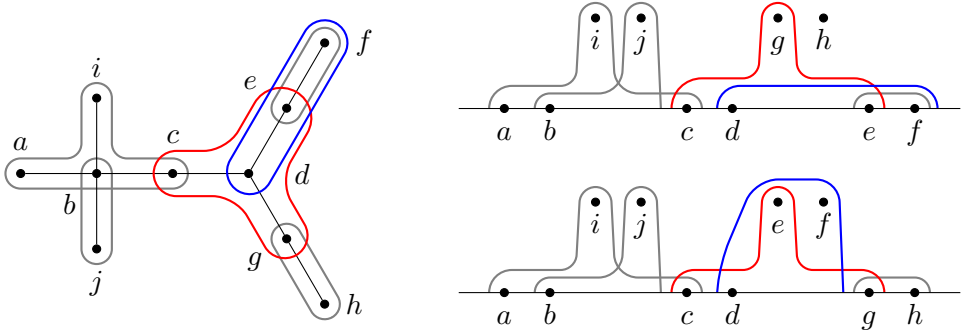
\begin{figure}[t]
\centering
\begin{minipage}{5cm}
\centering
\begin{tikzpicture}
  \tikzstyle{every node}=[circle,minimum size=3pt,inner sep=0pt,draw,fill]
  \node[label={[yshift=0.2cm]above:$a$}] (at) at (-3,0) {};
  \node[label={[xshift=-0.2cm,yshift=-0.25cm]below left:$b$}] (bt) at (-2,0) {};
  \node[label={[yshift=0.3cm]above:$c$}] (ct) at (-1,0) {};
  \node[label={[xshift=0.47cm]right:$d$}] (dt) at (0,0) {};
  \node[label={[shift=(150:0.4cm)]above left:$e$}] (et) at (60:1cm) {};
  \node[label={[xshift=0.25cm]right:$f$}] (ft) at (60:2cm) {};
  \node[label={[shift=(-150:0.3cm)]below left:$g$}] (gt) at (-60:1cm) {};
  \node[label={[xshift=0.15cm]right:$h$}] (ht) at (-60:2cm) {};
  \node[label={[yshift=0.2cm]above:$i$}] (it) at (-2,1) {};
  \node[label={[yshift=-0.2cm]below:$j$}] (jt) at (-2,-1) {};
  \draw[gray,thick,rounded corners=0.2cm] (-3.2,-0.2)--(-3.2,0.2)--(-2.2,0.2)--(-2.2,1.2)--(-1.8,1.2)--(-1.8,0.2)--(-0.8,0.2)--(-0.8,-0.2)--cycle;
  \draw[gray,thick,rounded corners=0.2cm] (-2.2,-1.2) rectangle (-1.8,0.2);
  \draw[gray,thick,rounded corners=0.2cm] {[shift=(60:0.8cm)](-30:0.2cm)--(150:0.2cm)}{[shift=(60:2.2cm)]--(150:0.2cm)--(-30:0.2cm)}--cycle;
  \draw[gray,thick,rounded corners=0.2cm] {[shift=(-60:0.8cm)](-150:0.2cm)--(30:0.2cm)}{[shift=(-60:2.2cm)]--(30:0.2cm)--(-150:0.2cm)}--cycle;
  \draw (at)--(bt)--(ct)--(dt)--(et)--(ft);
  \draw (dt)--(gt)--(ht);
  \draw (bt)--(it);
  \draw (bt)--(jt);
  \draw[red,thick,rounded corners=0.3cm] (-1.25,-0.3)--(-1.25,0.3)--(-0.2887,0.3){[shift=(60:1.25cm),rounded corners=0.4cm]--(150:0.4cm)--(-30:0.4cm)}--(0.4041,-0.1){[shift=(-60:1.25cm)]--(30:0.3cm)--(-150:0.3cm)}--(-0.1732,-0.3)--cycle;
  \draw[blue,thick,rounded corners=0.3cm] {[shift=(-120:0.27cm)](-30:0.3cm)--(150:0.3cm)}{[shift=(60:2.3cm)]--(150:0.3cm)--(-30:0.3cm)}--cycle;
\end{tikzpicture}
\end{minipage}\hskip 1cm
\begin{minipage}{6.6cm}
\centering
\begin{tikzpicture}[scale=0.6]
  \tikzstyle{every node}=[circle,minimum size=3pt,inner sep=0pt,draw,fill]
  \tikzstyle{every label}=[rectangle,draw=none,fill=none]
  \node[label=below:\strut $a$] (a1) at (0,0) {};
  \node[label=below:\strut $b$] (b1) at (1,0) {};
  \node[label=below:\strut $i$] (i1) at (2,2) {};
  \node[label=below:\strut $j$] (j1) at (3,2) {};
  \node[label=below:\strut $c$] (c1) at (4,0) {};
  \node[label=below:\strut $d$] (d1) at (5,0) {};
  \node[label=below:\strut $g$] (g1) at (6,2) {};
  \node[label=below:\strut $h$] (h1) at (7,2) {};
  \node[label=below:\strut $e$] (e1) at (8,0) {};
  \node[label=below:\strut $f$] (f1) at (9,0) {};
  \draw[gray,thick,rounded corners=0.2cm] (-0.333,0){[rounded corners=0.3cm]--(-0.333,0.5)}--(1.567,0.5)--(1.667,2.333)--(2.333,2.333){[rounded corners=0.3cm]--(2.433,0.333)}--(4.333,0.333)--(4.333,0);
  \draw[gray,thick,rounded corners=0.2cm] (0.667,0)--(0.667,0.333){[rounded corners=0.3cm]--(2.567,0.333)}--(2.667,2.333)--(3.333,2.333)--(3.433,0);
  \draw[gray,thick,rounded corners=0.2cm] (7.667,0)--(7.667,0.333)--(9.333,0.333)--(9.333,0);
  \draw[red,thick,rounded corners=0.2cm] (3.667,0){[rounded corners=0.4cm]--(3.667,0.667)}--(5.567,0.667)--(5.667,2.333)--(6.333,2.333)--(6.433,0.667){[rounded corners=0.4cm]--(8.333,0.667)}--(8.333,0);
  \draw[blue,thick,rounded corners=0.3cm] (4.667,0)--(4.667,0.5)--(9.5,0.5)--(9.5,0);
  \draw (-1,0)--(10,0);
\end{tikzpicture}\\[2ex]
\begin{tikzpicture}[scale=0.6]
  \tikzstyle{every node}=[circle,minimum size=3pt,inner sep=0pt,draw,fill];
  \tikzstyle{every label}=[rectangle,draw=none,fill=none];
  \node[label=below:\strut $a$] (a1) at (0,0) {};
  \node[label=below:\strut $b$] (b1) at (1,0) {};
  \node[label=below:\strut $i$] (i1) at (2,2) {};
  \node[label=below:\strut $j$] (j1) at (3,2) {};
  \node[label=below:\strut $c$] (c1) at (4,0) {};
  \node[label=below:\strut $d$] (d1) at (5,0) {};
  \node[label=below:\strut $e$] (g1) at (6,2) {};
  \node[label=below:\strut $f$] (h1) at (7,2) {};
  \node[label=below:\strut $g$] (e1) at (8,0) {};
  \node[label=below:\strut $h$] (f1) at (9,0) {};
  \draw[gray,thick,rounded corners=0.2cm] (-0.333,0){[rounded corners=0.3cm]--(-0.333,0.5)}--(1.567,0.5)--(1.667,2.333)--(2.333,2.333){[rounded corners=0.3cm]--(2.433,0.333)}--(4.333,0.333)--(4.333,0);
  \draw[gray,thick,rounded corners=0.2cm] (0.667,0)--(0.667,0.333){[rounded corners=0.3cm]--(2.567,0.333)}--(2.667,2.333)--(3.333,2.333)--(3.433,0);
  \draw[gray,thick,rounded corners=0.2cm] (7.667,0)--(7.667,0.333)--(9.333,0.333)--(9.333,0);
  \draw[red,thick,rounded corners=0.2cm] (3.667,0){[rounded corners=0.3cm]--(3.667,0.5)}--(5.567,0.5)--(5.667,2.333)--(6.333,2.333)--(6.433,0.5){[rounded corners=0.3cm]--(8.333,0.5)}--(8.333,0);
  \draw[blue,thick,rounded corners=0.3cm] (4.667,0)--(4.733,0.667)--(5.5,2.5)--(7.333,2.5)--(7.433,0);
  \draw (-1,0)--(10,0);
\end{tikzpicture}
\end{minipage}
\caption[]{A subtree overlap graph (left) and interval filament representations of its subgraphs induced on the subtrees intersecting $abcdef$ (top right) and $abcdgh$ (bottom right).
The domains of the interval filaments representing the subtrees $cdeg$ and $def$ overlap in the scenario $abcdef$ but are nested in the scenario $abcdgh$.}
\label{fig:subtree-game}
\end{figure}

In view of the above, a natural attempt is to define the on-line game corresponding to clean subtree overlap graphs just like the game $\gamIOV(k)$ but with representation by interval filaments instead of intervals.
However, this is not correct for the following reason.
We want to color the clean subtree overlap graph $G$ properly using the on-line approach of Lemma \ref{lem:game-upper}.
For each path $P$ in $F$ starting at a root, we will simulate an on-line algorithm on $G[V(P)]$ presenting the vertices in their order along $P$.
This way, we will present an interval filament graph.
The on-line approach will work correctly if the algorithm always assigns the same color to each vertex $x\in V(G)$, regardless of the choice of $P$.
This will be the case when the presentation scenarios up to the point when $u$ is presented are identical for all paths passing through $x$.
However, this cannot be guaranteed using the model of $G[V(P)]$ by interval filaments described above.
For example, for some two adjacent vertices $x,y\in V(G)$ lying on the common part of two paths $P_1$ and $P_2$, we may need to represent $x$ and $y$ by interval filaments whose domains are nested if we continue along $P_1$ but overlap if we continue along $P_2$.
See Figure \ref{fig:subtree-game} for such an example.
If the algorithm makes use of the representation, then the colorings it generates on $P_1$ and $P_2$ may be inconsistent.

To overcome the difficulty explained above, we provide a more abstract description of $G$, which we then use to define the on-line game.
For distinct vertices $x,y\in V(G)$, let $x\prec y$ denote that $x$ is an ancestor of $y$ in $F$.
We define relations $\includes$, $\overlaps$ and $\parallel$ on $V(G)$ as follows:
\begin{itemize}
\item $x\includes y$ if $x\prec y$ and the subtree $S_x$ contains the subtree $S_y$,
\item $x\overlaps y$ if $x\prec y$ and the subtrees $S_x$ and $S_y$ overlap,
\item $x\parallel y$ if $x\prec y$ and the subtrees $S_x$ and $S_y$ are disjoint.
\end{itemize}
It follows that the relations $\includes$, $\overlaps$ and $\parallel$ partition the relation $\prec$, that is, they are pairwise disjoint sets of pairs and their union gives the entire $\prec$.
Furthermore, the following conditions are satisfied for any $x,y,z\in V(G)$:
\begin{enumeratex}{A}
\item\label{item:A1} if $x\includes y$ and $y\includes z$, then $x\includes z$,
\item\label{item:A2} if $x\includes y$ and $y\overlaps z$, then $x\includes z$ or $x\overlaps z$,
\item\label{item:A3} if $x\overlaps y$ and $y\includes z$, then $x\overlaps z$ or $x\parallel z$ (because of cleanness),
\item\label{item:A4} if $x\parallel y$ and $y\prec z$, then $x\parallel z$.
\end{enumeratex}
We define an on-line game $\gamABS(k)$ in which Presenter builds a graph $G$ together with relations $\includes$, $\overlaps$ and $\parallel$ declaring, in each round, the relations $\includes$, $\overlaps$ and $\parallel$ between the vertices presented before and the new vertex so that
\begin{enumerate}
\item\label{item:abs-relations} $\includes$, $\overlaps$ and $\parallel$ partition the order of presentation $\prec$ and satisfy \ref{item:A1}--\ref{item:A4},
\item\label{item:abs-graph} $xy\in E(G)$ if and only if $x\overlaps y$ or $y\overlaps x$,
\item\label{item:abs-omega} $\omega(G)\leq k$,
\end{enumerate}
and Algorithm properly colors $G$ on-line.
No interval filament intersection model is revealed by Presenter in the game $\gamABS(k)$.
See Figure \ref{fig:abs-filament} for an illustration of possible representations of the relations $\includes$, $\overlaps$ and $\parallel$ in the game.

\begin{figure}[t]
\centering
\begin{minipage}[b]{2.4cm}
\centering
\begin{tikzpicture}
  \draw (0,0)--(2.4,0);
  \draw[thick,rounded corners=0.5cm] (0.4,0)--(0.7,1.5)--(1.7,1.5)--(2,0);
  \draw[thick,rounded corners=0.2cm] (0.8,0)--(1,1)--(1.4,1)--(1.6,0);
  \tikzstyle{every node}=[below,inner sep=1pt]
  \node at (0.4,0) {\strut $a$};
  \node at (0.8,0) {\strut $b$};
\end{tikzpicture}\\
$a\includes b$
\end{minipage}\hskip 1cm
\begin{minipage}[b]{5.2cm}
\centering
\begin{tikzpicture}
  \draw (0,0)--(2.4,0);
  \draw[thick,rounded corners=0.3cm] (0.4,0)--(0.6,1)--(1.8,1)--(2,0);
  \draw[thick,rounded corners=0.2cm] (0.9,0)--(1,1.5)--(1.4,1.5)--(1.5,0);
  \tikzstyle{every node}=[below,inner sep=1pt]
  \node at (0.4,0) {\strut $a$};
  \node at (0.9,0) {\strut $b$};
\end{tikzpicture}\nobreak\hskip 0.4cm
\begin{tikzpicture}
  \draw (0,0)--(2.4,0);
  \draw[thick,rounded corners=5pt] (0.4,0)--(0.6,1.5)--(1,1.5)--(1.5,0);
  \draw[thick,rounded corners=5pt] (0.9,0)--(1.4,1.5)--(1.8,1.5)--(2,0);
  \tikzstyle{every node}=[below,inner sep=1pt]
  \node at (0.4,0) {\strut $a$};
  \node at (0.9,0) {\strut $b$};
\end{tikzpicture}\\
$a\overlaps b$
\end{minipage}\hskip 1cm
\begin{minipage}[b]{2.4cm}
\centering
\begin{tikzpicture}
  \draw (0,0)--(2.4,0);
  \draw[thick,rounded corners=0.2cm] (0.4,0)--(0.5,1.5)--(0.9,1.5)--(1,0);
  \draw[thick,rounded corners=0.2cm] (1.4,0)--(1.5,1.5)--(1.9,1.5)--(2,0);
  \tikzstyle{every node}=[below,inner sep=1pt]
  \node at (0.4,0) {\strut $a$};
  \node at (1.4,0) {\strut $b$};
\end{tikzpicture}\\
$a\parallel b$
\end{minipage}
\caption[]{Interval filament representations of $a\includes b$, $a\overlaps b$, and $a\parallel b$.
The two drawings of $a\overlaps b$ distinguish whether the domains of the interval filaments representing $a$ and $b$ are nested or overlap.
Other drawings of $a\overlaps b$ can be obtained by letting $a$ and $b$ cross many times.}
\label{fig:abs-filament}
\end{figure}
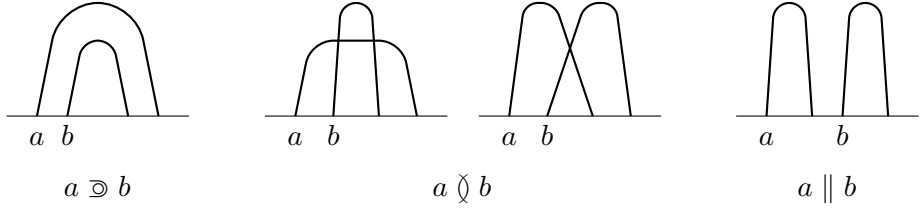

\begin{lemma}
\label{lem:subtree-game}
A graph\/ $G$ is a game graph of\/ $\gamABS(k)$ if and only if\/ $G$ is a clean subtree overlap graph.
If\/ $G$ is a game graph of\/ $\gamABS(k)$ and the relation\/ $\prec$ (as defined for a game graph) is a total order on\/ $V(G)$, then\/ $G$ is an interval filament graph.
\end{lemma}

\begin{proof}
We have argued above that every clean subtree overlap graph with clique number at most $k$ is a game graph of $\gamABS(k)$.
Now, suppose that $G$ is a game graph of $\gamABS(k)$.
This means that there exist a rooted forest $F$ on $V(G)$ and relations $\includes$, $\overlaps$ and $\parallel$ on $V(G)$ such that
\begin{enumerate}
\item $\includes$, $\overlaps$ and $\parallel$ partition the ancestor-descendant order $\prec$ of $F$ and satisfy \ref{item:A1}--\ref{item:A4},
\item $xy\in E(G)$ if and only if $x\overlaps y$ or $y\overlaps x$,
\item $\omega(G)\leq k$.
\end{enumerate}
Let $T$ be a tree with
\begin{align*}
V(T)&=\{r\}\cup\{u_x\colon x\in V(G)\}\cup\{v_x\colon x\in V(G)\},\\
E(T)&=\{ru_x\colon\text{$x$ is a root of $F$}\}\cup\{u_xu_y\colon xy\in E(F)\}\cup\{u_xv_x\colon x\in V(G)\}.
\end{align*}
For $x\in V(G)$, let $S_x=\{u_x,v_x\}\cup\{u_y\colon x\includes y$ or $x\overlaps y\}\cup\{v_y\colon x\includes y\}$.
We show that $x\mapsto S_x$ is a clean overlap model of $G$ by subtrees of $T$.

If $x\prec y$ and $u_y\notin S_x$, then $x\parallel y$, so it follows from \ref{item:A4} that $y\prec z$ implies $x\parallel z$ and thus $u_z\notin S_x$ for every $z$.
Hence, every set $S_x$ is the node set of a subtree of $T$.
If $x\includes y$, then $y\includes z$ implies $x\includes z$, by \ref{item:A1}, and $y\overlaps z$ implies $x\includes z$ or $x\overlaps z$, by \ref{item:A2}, and hence $S_y\subset S_x$.
If $x\overlaps y$, then $u_x,v_x\in S_x\setminus S_y$, $u_y\in S_x\cap S_y$, and $v_y\in S_y\setminus S_x$, and hence $S_x$ and $S_y$ overlap.
Finally, if $x\parallel y$, then it follows from \ref{item:A4} that $y\prec z$ implies $x\parallel z$ for every $z$, and hence $S_x\cap S_y=\emptyset$.
This shows that $x\mapsto S_x$ is indeed an overlap model of $G$.
Moreover, by \ref{item:A3}, there are no $x$, $y$, $z$ with $x\overlaps y\includes z$ and $x\includes z$, so the model is clean.
This completes the proof of the first statement.

For the proof of the second statement, assume that the underlying forest $F$ of the game graph $G$ consists of just one root-to-leaf path.
It follows directly from the construction that all sets $S_x$ for $x\in V(G)$ intersect the set $\{u_x\colon x\in V(G)\}$, which forms a path in $T$.
As it has been explained earlier in this section, an overlap graph of subtrees of $T$ all of which intersect some path in $T$ is an interval filament graph.
\end{proof}

The game $\gamIOV(k)$ is more restrictive for Presenter than the game $\gamABS(k)$, in the sense that every presentation scenario in the former can be translated into a presentation scenario in the latter.
Indeed, let $G$ be a graph presented in $\gamIOV(k)$ together with its representation $\mu\colon V(G)\to\calI$, and let $\prec$ be its order of presentation.
We can define relations $\includes$, $\overlaps$ and $\parallel$ on $V(G)$ just like before:
\begin{itemize}
\item $x\includes y$ if $x\prec y$ and the interval $\mu(x)$ contains the interval $\mu(y)$,
\item $x\overlaps y$ if $x\prec y$ and the intervals $\mu(x)$ and $\mu(y)$ overlap,
\item $x\parallel y$ if $x\prec y$ and the intervals $\mu(x)$ and $\mu(y)$ are disjoint.
\end{itemize}
Clearly, the relations $\includes$, $\overlaps$ and $\parallel$ thus defined satisfy the conditions \ref{item:abs-relations}--\ref{item:abs-omega} of the definition of $\gamABS(k)$.
This and Lemma \ref{lem:subtree-game} imply that every clean interval overlap game graph is a clean subtree overlap graph.

\section{Coloring algorithm for rectangle and subtree overlap graphs}
\label{sec:coloring}

In this section, we will prove that game graphs of $\gamABS(k)$ have chromatic number $O_k((\log\log n)^{k-1})$, while game graphs of $\gamIOV(k)$ (which are the same as game graphs of $\gamIOV_3(k)$) have chromatic number $O_k(\log\log n)$ and triangle-free chromatic number $O_k(1)$.
Then, the same bounds on the chromatic number of clean subtree overlap graphs and (respectively) the chromatic number and triangle-free chromatic number of rectangle overlap graphs will follow from Lemmas \ref{lem:subtree-game} and \ref{lem:rectangle-to-game} (respectively).

The general idea is to provide on-line algorithms in $\gamABS(k)$, $\gamIOV(k)$ and $\gamIOV_3(k)$ using few colors, and then to use Lemma \ref{lem:game-upper} to derive upper bounds on the (triangle-free) chromatic number of their game graphs.
However, since Presenter has a strategy to force Algorithm to use $\Omega(\log r)$ colors in $r$ rounds of the game $\gamIOV(2)$, a direct application of Lemma \ref{lem:game-upper} to the game graph cannot succeed for $\gamABS(k)$ and $\gamIOV(k)$ if the rooted forest $F$ underlying the game graph contains long paths.
To overcome this problem, we use the technique of \emph{heavy-light decomposition} due to Sleator and Tarjan \cite{ST83}.

Let $G$ be a game graph of $\gamABS(k)$ or $\gamIOV(k)$ with $n$ vertices and with an underlying forest $F$.
Thus $\omega(G)\leq k$.
We call an edge $xy$ of $F$, where $y$ is a child of $x$, \emph{heavy} if the subtree of $F$ rooted at $y$ contains more than half of the vertices of the subtree of $F$ rooted at $x$, and we call it \emph{light} otherwise.
The following is proved by an easy induction.

\begin{lemma}[Sleator, Tarjan \cite{ST83}]
\label{lem:heavy-paths}
Every path in\/ $F$ from a root to a leaf contains at most\/ $\lfloor\log_2n\rfloor$ light edges.
\end{lemma}

Every vertex of $F$ has a heavy edge to at most one of its children, so the heavy edges form a collection of paths in $F$, called \emph{heavy paths}.
For each heavy path $P$, by the second statement of Lemma \ref{lem:subtree-game}, the graph $G[V(P)]$ is an interval filament graph, and therefore, by Theorem \ref{thm:filament}~\ref{item:filament-upper}, it can be colored properly using $O_k(1)$ colors.
In the special case that $G$ is a game graph of $\gamIOV(k)$, the result of Kostochka and Milans \cite{KM12} implies that $2k-1$ colors even suffice.
We start with a preliminary coloring of the vertices of $G$ that colors each heavy path as it is described above using the same set of colors, thus using $O_k(1)$ colors in total.
Note that this is not an on-line coloring---the color of a vertex depends on the subgraph induced on the whole heavy path that contains it.

Let $b=\lfloor\log_2n\rfloor+1$.
By Lemma \ref{lem:heavy-paths}, every root-to-leaf path in $F$ is subdivided by its light edges into at most $b$ blocks, each being a subpath of some heavy path of $F$.
The subgraph of $G$ induced on each color class in the preliminary coloring is itself a game graph of $\gamABS(k)$ or $\gamIOV(k)$, respectively, and contains no edges within any of the blocks.
We will color each such subgraph separately by an appropriate on-line algorithm using $O_k((\log b)^{k-1})$ colors in $\gamABS(k)$ and $O_k(\log b)$ colors in $\gamIOV(k)$.
To achieve this formally, we define on-line games $\gamABS(k,b)$ and $\gamIOV(k,b)$ like $\gamABS(k)$ and $\gamIOV(k)$, respectively, but with one additional constraint:
\begin{enumeratefixed}{iv/v}
\item there is a partition of the vertices into at most $b$ blocks of vertices consecutive in the order of presentation $\prec$ such that no edge connects vertices in the same block.
\end{enumeratefixed}
It follows from the discussion above that the subgraph of $G$ induced on each color class in the preliminary coloring is a game graph of $\gamABS(k,b)$ or $\gamIOV(k,b)$, respectively.

The rest of this section is devoted to the proofs of the following three lemmas.
In view of the discussion above, combining them with the results of previous sections will then give us Theorems \ref{thm:subtree} \ref{item:subtree-upper}--\ref{item:clean-subtree-upper}, \ref{thm:rectangle} \ref{item:rectangle-upper}--\ref{item:clean-rectangle-upper}, and \ref{thm:rectangle-K3}.

\begin{lemma}
\label{lem:clean-subtree-coloring}
There is an on-line\/ $O_k((\log b)^{k-1})$-coloring algorithm in\/ $\gamABS(k,b)$.
\end{lemma}

\begin{lemma}
\label{lem:overlap-game-coloring}
There is an on-line\/ $O_k(\log b)$-coloring algorithm in\/ $\gamIOV(k,b)$.
\end{lemma}

\begin{lemma}
\label{lem:overlap-game-K3-coloring}
There is an on-line\/ $O_k(1)$-coloring algorithm in\/ $\gamIOV_3(k)$.
\end{lemma}

For the next part of this section, we forget the preceding context (in particular, the previous meaning of $G$) and adopt the setting of Lemma \ref{lem:clean-subtree-coloring}: a graph $G$ with relations $\includes$, $\overlaps$ and $\parallel$ is being presented in the game $\gamABS(k,b)$, and we are to color $G$ properly using $O_k((\log b)^{k-1})$ colors on-line.
Whatever we show for $\gamABS(k,b)$ applies also to $\gamIOV(k,b)$, as the latter is more restrictive for Presenter.
The proof of Lemma \ref{lem:overlap-game-coloring} will differ from the proof of Lemma \ref{lem:clean-subtree-coloring} only in one part, where the use of a direct argument instead of induction will allow us to reduce the number of colors to $O_k(\log b)$.
The last part of that proof, which raises the number of colors from $O_k(1)$ to $O_k(\log b)$, can be omitted when we aim only at a triangle-free coloring, whence Lemma \ref{lem:overlap-game-K3-coloring} will follow.

As a new vertex $z$ of $G$ is presented, we classify it as \emph{primary} or \emph{secondary} according to the following on-line rule: if there are $x,y\in V(G)$ such that $y$ is primary, $x\overlaps y$, $x\overlaps z$, and $y\includes z$, then $z$ is secondary, otherwise $z$ is primary.
Let $P$ denote the set of primary vertices, being built on-line during the game.
For every $y\in P$, let $S(y)$ be the set containing $y$ and all secondary vertices $z$ such that $y\includes z$ and there is $x$ with $x\overlaps y$ and $x\overlaps z$, also being built on-line during the game.
See Figure \ref{fig:abs-scenario} for an illustration.
The following lemma will be used implicitly throughout the rest of this section.

\begin{figure}[t]
\begin{tikzpicture}[scale=0.5]
  \draw (0,0)--(22.2,0);
  \draw[thick,rounded corners=0.65cm] (0.8,0)--(1.6,4)--(4.2,4)--(4.8,1)[rounded corners=0.3cm]--(12,1)--(12.2,0);
  \draw[thick,rounded corners=0.35cm] (1.8,0)--(2.2,2)--(3.6,2)--(4,0);
  \draw[thick,rounded corners=0.6cm] (2.8,0)--(3.4,3)[rounded corners=0.3cm]--(14.2,3)[rounded corners=0.2cm]--(15.9,2)[rounded corners=0.5cm]--(19,2)--(19.4,0);
  \draw[thick,rounded corners=0.4cm] (5,0)--(5.4,2)[rounded corners=0.6cm]--(12.8,2)--(13.2,0);
  \draw[thick,rounded corners=0.55cm] (6,0)--(7,5)--(9.2,5)--(10.2,0);
  \draw[thick,rounded corners=0.15cm] (7,0)--(7.8,4)--(8.4,4)--(9.2,0);
  \draw[thick,rounded corners=0.3cm] (11.2,0)--(11.4,1){[rounded corners=0.2cm]--(13.1,2)--(13.9,2)}--(15.6,1)--(16.2,0);
  \draw[thick,rounded corners=0.5cm] (14.2,0)--(15,4)--(20.6,4)--(21.4,0);
  \draw[thick,rounded corners=0.2cm] (15.2,0)--(15.8,3)--(16.6,3)--(18.4,0);
  \draw[thick,rounded corners=0.2cm] (17.2,0)--(19,3)--(19.8,3)--(20.4,0);
  \tikzstyle{every node}=[below,inner sep=1pt]
  \node at (0.8,0) {\strut $a$};
  \node at (1.8,0) {\strut $b$};
  \node at (2.8,0) {\strut $c$};
  \node at (5,0) {\strut $d$};
  \node at (6,0) {\strut $e$};
  \node at (7,0) {\strut $f$};
  \node at (11.2,0) {\strut $g$};
  \node at (14.2,0) {\strut $h$};
  \node at (15.2,0) {\strut $i$};
  \node at (17.2,0) {\strut $j$};
\end{tikzpicture}
\caption[]{A presentation scenario of an interval filament graph in the game $\gamABS(3)$ and one of its possible representations.
The representation is for illustration only and is not revealed by Presenter in the game.
The primary vertices are $a,b,c,e,h$.
We have $S(a)=\{a\}$, $S(b)=\{b\}$, $S(c)=\{c,d,g\}$, $S(e)=\{e,f\}$, and $S(h)=\{h,i,j\}$.}
\label{fig:abs-scenario}
\end{figure}
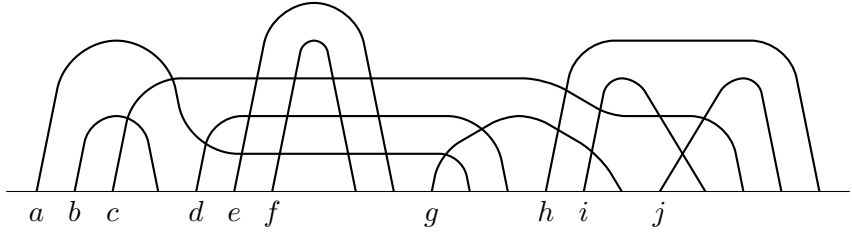

\begin{lemma}
For every\/ $z\in V(G)$, there is a unique vertex\/ $p\in P$ with\/ $z\in S(p)$.
\end{lemma}

\begin{proof}
It follows directly from the primary-secondary classification rule that there is $p\in P$ such that $z\in S(p)$.
To see that such a vertex $p$ is unique, suppose to the contrary that there are $p,q\in P$ such that $p\prec q$ and $z\in S(p)\cap S(q)$.
It follows that $z$ is secondary, $p\includes z$, and $q\includes z$.
We can have neither $p\overlaps q$, as this would contradict \ref{item:A3}, nor $p\parallel q$, as this would contradict \ref{item:A4}.
Hence $p\includes q$.
Since $z\in S(p)$, there is $x\in V(G)$ such that $x\overlaps p$ and $x\overlaps z$.
Since $x\overlaps p\includes q$, we have $x\overlaps q$ or $x\parallel q$, by \ref{item:A3}.
However, we cannot have $x\parallel q$, as this and $q\includes z$ would contradict \ref{item:A4}.
Hence $x\overlaps q$.
This contradicts the hypothesis that $q$ is primary.
\end{proof}

The next lemma will allow us to construct an on-line coloring of $G$ from on-line colorings of $G[P]$ and of all $G[S(p)]$ with $p\in P$.

\begin{lemma}
\label{lem:secondary-coloring}
The vertices in\/ $P$ can be\/ $2$-colored on-line so that if\/ $p,q\in P$ have the same color and\/ $pq\notin E(G[P])$, then $xy\notin E(G)$ for any\/ $x\in S(p)$ and\/ $y\in S(q)$.
\end{lemma}

\begin{proof}
We make the following two observations:
\begin{enumerate}
\item\label{item:sec1} If $p,q\in P$, $p\prec q$, $pq\notin E(G)$, $x\in S(p)$, $y\in S(q)$, and $xy\in E(G)$, then $x\overlaps q$.
\item\label{item:sec2} For every $q\in P$, there is at most one vertex $p\in P$ with the following properties: $p\prec q$, $pq\notin E(G)$, and there is $x\in S(p)$ with $x\overlaps q$.
\end{enumerate}
Once they are established, we can argue as follows.
By \ref{item:sec2}, $P$ can be colored on-line using two colors so as to distinguish any $p,q\in P$ such that $p\prec q$, $pq\notin E(G)$, and there is $x\in S(p)$ with $x\overlaps q$.
It follows from \ref{item:sec1} that if $p,q\in P$, $p\prec q$, $pq\notin E(G)$, $x\in S(P)$, $y\in S(q)$, and $xy\in E(G)$, then $x\overlaps q$ and therefore $p$ and $q$ have distinct colors.

It remains to prove \ref{item:sec1} and \ref{item:sec2}.
First, we prove the following property:
\begin{enumerate}
\setcounter{enumi}{2}
\item\label{item:sec3} If $p,q\in P$, $p\prec q$, $pq\notin E(G)$, $x\in S(p)$, $q\includes y$, and $xy\in E(G)$, then $p\includes q$ and $x\overlaps q$.
\end{enumerate}

Suppose $p\parallel q$.
We cannot have $q\prec x$, as this would imply $p\parallel x$, by \ref{item:A4}.
Hence $x\prec q$.
It follows from \ref{item:A4} that $p\parallel y$.
But $p\includes x$ and $x\overlaps y$ imply $p\includes y$ or $p\overlaps y$, by \ref{item:A2}, thus contradicting $p\parallel y$.
Therefore, we cannot have $p\parallel q$.
We cannot have $p\overlaps q$ either, as $pq\notin E(G)$.
So we have $p\includes q$.
Since $x\in S(p)$, there is some $u$ with $u\overlaps p$ and $u\overlaps x$.
We cannot have $u\includes q$, because this would contradict \ref{item:A3}.
We cannot have $u\overlaps q$, because then $q$ would be secondary.
Hence $u\parallel q$.
This implies $x\prec q$, whence we have $p\includes x\prec q\includes y$ and $x\overlaps y$.
We cannot have $x\includes q$, because this would imply $x\includes y$.
We cannot have $x\parallel q$ either, by \ref{item:A4}.
Hence $x\overlaps q$.

Now, \ref{item:sec1} follows immediately from \ref{item:sec3}.
To see \ref{item:sec2}, suppose there are $p_1,p_2,q\in P$ such that $p_1\prec p_2\prec q$, $p_1q\notin E(G)$, $p_2q\notin E(G)$, and there are $x_1\in S(p_1)$ and $x_2\in S(p_2)$ with $x_1\overlaps q$ and $x_2\overlaps q$.
By \ref{item:sec3}, we have $p_1\includes q$ and $p_2\includes q$.
We cannot have $p_1\overlaps p_2$, as this would contradict \ref{item:A3} for $p_1$, $p_2$ and $q$.
Hence $p_1p_2\notin E(G)$.
We apply \ref{item:sec3} to $p_1$, $p_2$, $x_1$ and $q$ to conclude that $x_1\overlaps p_2$.
Now, since $x_1\overlaps p_2\includes q$ and $x_1\overlaps q$, we conclude that $q$ is secondary, which is a contradiction.
\end{proof}

The following lemma will allow us to color $G[S(p)]$ for every $p\in P$.

\begin{lemma}
\label{lem:secondary-common}
For every\/ $p\in P$, there is\/ $x\in V(G)$ with\/ $x\overlaps y$ for all\/ $y\in S(p)$.
\end{lemma}

\begin{proof}
Let $p\in P$.
Let $z$ be the latest presented vertex in $S(p)$.
It follows that there is $x\in V(G)$ such that $x\overlaps p$ and $x\overlaps z$.
Now, take any $y\in S(p)\setminus\{z\}$.
We have $x\overlaps p$ and $p\includes y$, so $x\overlaps y$ or $x\parallel y$, by \ref{item:A3}.
We cannot have $x\parallel y$, as this would imply $x\parallel z$, by \ref{item:A4}.
Hence $x\overlaps y$.
\end{proof}

It follows from Lemma \ref{lem:secondary-common} that $\omega(G[S(p)])\leq k-1$ for every $p\in P$.
This will allow us to use induction to color every $G[S(p)]$ in the abstract overlap game.
For the clean interval overlap game, instead of induction, we will use the following direct argument.

\begin{lemma}
\label{lem:overlap-game-secondary}
If\/ $G$ is a clean interval overlap graph presented on-line in the game\/ $\gamIOV(k,b)$ or\/ $\gamIOV_3(k)$, then, for every\/ $p\in P$, the graph\/ $G[S(p)]$ can be properly colored on-line using at most\/ $\smash{\binom{k}{2}}$ colors.
\end{lemma}

\begin{proof}
Let $\mu$ denote the clean interval overlap representation of $G$ presented in the game together with $G$.
Consider one of the sets $S(p)$ being built during the game.
By Lemma \ref{lem:secondary-common}, there is $x\in V(G)$ such that $\leftside(\mu(x))<\leftside(\mu(y))<\rightside(\mu(x))<\rightside(\mu(y))$ for every $y\in S(p)$.
Define a partial order $<$ on $S(p)$ so that $y<z$ whenever $\leftside(\mu(y))<\leftside(\mu(z))$ and $\rightside(\mu(y))>\rightside(\mu(z))$.
It follows that $G[S(p)]$ is the incomparability graph of $S(p)$ with respect to $<$.
Moreover, the set $S(p)$ is built in the up-growing manner with respect to $<$, that is, every vertex is maximal with respect to $<$ at the moment it is presented.
Since $\omega(G[S(p)])\leq k-1$, it follows from Theorem \ref{thm:coco-game} that the graph $G[S(p)]$ can be properly colored on-line using $\smash{\binom{k}{2}}$ colors.
\end{proof}

To prove Lemmas \ref{lem:clean-subtree-coloring} and \ref{lem:overlap-game-coloring}, we will color the graph $G[P]$ in two steps, expressed by Lemmas \ref{lem:primary-coloring-1} and \ref{lem:primary-coloring-2}.
Only the first step is needed for the proof of Lemma \ref{lem:overlap-game-K3-coloring}.

\begin{lemma}
\label{lem:primary-coloring-1}
The graph\/ $G[P]$ can be colored on-line using\/ $k$ colors so that the following holds for any\/ $x,y,z\in P$ of the same color:
\begin{equation*}
\label{eq:col-prop}
\text{if\/ $x\overlaps y\prec z$, then\/ $x\parallel z$ or\/ $y\parallel z$;}\tag{$*$}
\end{equation*}
in particular, the coloring of\/ $G[P]$ is triangle-free.
\end{lemma}

\begin{proof}
We use the following two observations:
\begin{enumerate}
\item\label{item:col-prop-1} If $x,y,z$ do not satisfy \eqref{eq:col-prop}, then neither do $x,y,y'$ for any $y'$ with $y\prec y'\prec z$.
\item\label{item:col-prop-2} If $x,y,z$ are in $P$ and do not satisfy \eqref{eq:col-prop}, then $y\overlaps z$.
\end{enumerate}
To see \ref{item:col-prop-1}, suppose that $x\overlaps y\prec y'\prec z$ and $x,y,y'$ satisfy \eqref{eq:col-prop}, that is, $x\parallel y'$ or $y\parallel y'$.
By \ref{item:A4}, this yields $x\parallel z$ or $y\parallel z$, respectively, so $x,y,z$ satisfy \eqref{eq:col-prop}.
To see \ref{item:col-prop-2}, suppose $x\overlaps y\includes z$.
By \ref{item:A3}, this yields $x\overlaps z$ or $x\parallel z$.
We cannot have $x\overlaps z$, as then $z$ would be secondary.
Hence $x\parallel z$, so $x,y,z$ satisfy \eqref{eq:col-prop}.

The coloring of $G[P]$ is constructed as follows.
At the time when a vertex $z\in P$ is presented, consider the set $Y$ of all vertices $y\in P$ for which there is $x\in P$ such that $x,y,z$ do not satisfy \eqref{eq:col-prop}.
By \ref{item:col-prop-1}, for any $y,y'\in Y\cup\{z\}$ with $y\prec y'$, there is $x\in P$ such that $x,y,y'$ do not satisfy \eqref{eq:col-prop}.
This and \ref{item:col-prop-2} imply that $Y\cup\{z\}$ is a clique in $G[P]$, and hence $|Y|\leq k-1$.
Therefore, at least one of the $k$ colors is not used on any vertex from $Y$, and we use such a color for $z$.
It is clear that the coloring of $G[P]$ thus obtained satisfies the condition of the lemma.
\end{proof}

\emph{First-fit} is the on-line algorithm that colors the graph properly with positive integers in a greedy way: when a new vertex $v$ is presented, it is assigned the least color that has not been used on any of the neighbors of $v$ presented before $v$.

\begin{theorem}[folklore]
\label{thm:first-fit}
First-fit uses at most\/ $\lfloor\log_2n\rfloor+1$ colors on any forest with\/ $n$ vertices presented in any order.
\end{theorem}

Let $P'$ be a subset of $P$ being built on-line during the game so that any $x,y,z\in P'$ satisfy the condition \eqref{eq:col-prop} of Lemma \ref{lem:primary-coloring-1}.
For the proofs of Lemmas \ref{lem:clean-subtree-coloring} and \ref{lem:overlap-game-coloring}, we apply First-fit to obtain a proper coloring of $G[P']$.

\begin{lemma}
\label{lem:primary-coloring-2}
First-fit colors the graph\/ $G[P']$ properly on-line using\/ $O(\log b)$ colors.
\end{lemma}

\begin{proof}
Let $R$ denote the set of vertices in $P'$ that have no neighbor to the right in $G[P']$.
We show that each member of $P'\setminus R$ has at most one neighbor to the right in $G[P'\setminus R]$.
Suppose to the contrary that there are $x,y,z\in P'\setminus R$ with $x\overlaps y\prec z$ and $x\overlaps z$.
Since $y\in P'\setminus R$, there is $z'\in P'$ such that $y\overlaps z'$.
Since $x\overlaps z$, we have $y\parallel z$, and since $y\overlaps z'$, we have $x\parallel z'$, because $x,y,z$ and $x,y,z'$ satisfy the condition \eqref{eq:col-prop} of Lemma \ref{lem:primary-coloring-1}.
However, we have $z\prec z'$ or $z'\prec z$, which implies either $y\parallel z'$ or $x\parallel z$, by \ref{item:A4}.
This contradiction shows that each member of $P'\setminus R$ has at most one neighbor to the right in $G[P'\setminus R]$.
In particular, $G[P'\setminus R]$ is a forest.

By the definition of $R$, the colors assigned by First-fit to the vertices in $P'\setminus R$ do not depend on the colors assigned to the vertices in $R$.
In particular, if we ran First-fit only on the graph $G[P'\setminus R]$, then we would obtain exactly the same colors on the vertices in $P'\setminus R$.
Let $a$ be the maximum color used by First-fit on $G[P']$.
Since there is a vertex in $P'$ with color $a$, there must be a vertex in $P'\setminus R$ with color $a-1$.
This, the fact that $G[P'\setminus R]$ is a forest, and Theorem \ref{thm:first-fit} yield $a\leq\lfloor\log_2{|P'|}\rfloor+2$.

We apply a similar reasoning to show $a\leq\lfloor\log_2b\rfloor+3$.
Recall the assumption that there is a partition of $V(G)$ into at most $b$ blocks of $\prec$-consecutive vertices such that no edge of $G[P']$ connects vertices in the same block.
Let $Q$ be the set obtained from $P'\setminus R$ by removing all vertices with color $1$.
If we ran First-fit only on $G[Q]$, then each vertex in $Q$ would get the color less by $1$ than the color it has received in the first-fit coloring of $G[P'\setminus R]$.
Therefore, our hypothetical run of First-fit on $G[Q]$ uses at least $a-2$ colors, which yields $a\leq\lfloor\log_2{|Q|}\rfloor+3$, by Theorem \ref{thm:first-fit}.
Now, it is enough to prove that each block $B$ of $\prec$-consecutive vertices of $G$ such that $G[B]$ has no edge can contain at most one vertex of $Q$, as this will imply $|Q|\leq b$.
Suppose to the contrary that there are two vertices $y_1,y_2\in Q\cap B$ with $y_1\prec y_2$.
By the assumption that $G[B]$ has no edge, we do not have $y_1\overlaps y_2$.
Each member of $Q$ has a neighbor to the left and a neighbor to the right in $G[P']$, neither of which can belong to $B$.
Therefore, there are $x,z\in P'$ such that $x\prec y_1\prec y_2\prec z$, $x\overlaps y_2$, and $y_1\overlaps z$.
We cannot have $y_1\parallel y_2$, as this and $y_2\prec z$ would imply $y_1\parallel z$, by \ref{item:A4}.
Hence $y_1\includes y_2$.
We cannot have $x\parallel y_1$, as this and $y_1\prec y_2$ would imply $x\parallel y_2$, by \ref{item:A4}.
Neither can we have $x\includes y_1$, as this and $y_1\includes y_2$ would imply $x\includes y_2$, by \ref{item:A1}.
Hence $x\overlaps y_1$.
This, $y_1\includes y_2$, and $x\overlaps y_2$ contradict the assumption that $y_2$ is primary.
We have thus shown $a=O(\log b)$, which completes the proof.
\end{proof}

\begin{proof}[Proof of Lemma \ref{lem:clean-subtree-coloring}]
The proof goes by induction on $k$.
The case $k=1$ is trivial.
Now, assume that $k\geq 2$ and the lemma holds for $k-1$.
By Lemma \ref{lem:primary-coloring-1}, $G[P]$ can be colored on-line using colors $1,\ldots,k$ so as to guarantee the condition \eqref{eq:col-prop} for any $x,y,z\in P$.
For $p\in P$, let $\phi(p)$ denote the color of $p$ in such a coloring.
For $i\in\{1,\ldots,k\}$, let $P_i=\{p\in P\colon\phi(p)=i\}$.
By Lemma \ref{lem:primary-coloring-2}, each $G[P_i]$ can be properly colored on-line using colors $1,\ldots,\ell$, where $\ell=O(\log b)$.
For $p\in P_i$, let $\psi(p)$ denote the color of $p$ in such a coloring.
For $i\in\{1,\ldots,k\}$ and $j\in\{1,\ldots,\ell\}$, let $P_{i,j}=\{p\in P_i\colon\psi(p)=j\}$.
By Lemma \ref{lem:secondary-coloring}, each set $P_{i,j}$ can be further $2$-colored on-line so as to distinguish any $p,q\in P_{i,j}$ for which there is some edge between $S(p)$ and $S(q)$.
Let $\zeta$ be such a $2$-coloring of each $P_{i,j}$ using colors $1$ and~$2$.
For each $p\in P$, it follows from Lemma \ref{lem:secondary-common} that $\omega(G[S(p)])\leq k-1$ and therefore, by the induction hypothesis, $G[S(p)]$ can be properly colored on-line using colors $1,\ldots,m$, where $m=O_k((\log b)^{k-2})$.
For $p\in P$ and $x\in S(p)$, let $\xi(x)$ denote the color of $x$ in such a coloring.
We color each vertex $x\in S(p)$ by the quadruple $(\phi(p),\psi(p),\zeta(p),\xi(x))$.
This is a proper coloring of $G$ using at most $2k\ell m=O_k((\log b)^{k-1})$ colors.
\end{proof}

\begin{proof}[Proof of Lemma \ref{lem:overlap-game-coloring}]
The proof goes as above with one change: for every $p\in P$, we apply Lemma \ref{lem:overlap-game-secondary} instead of induction to color $G[S(p)]$ properly using colors $1,\ldots,\smash{\binom{k}{2}}$.
This gives a proper coloring of $G$ using at most $2k\ell\smash{\binom{k}{2}}=O_k(\log b)$ colors.
\end{proof}

\begin{proof}[Proof of Lemma \ref{lem:overlap-game-K3-coloring}]
By Lemma \ref{lem:primary-coloring-1}, $G[P]$ can be triangle-free colored on-line using colors $1,\ldots,k$.
For $p\in P$, let $\phi(p)$ denote the color of $p$ in such a coloring.
For $i\in\{1,\ldots,k\}$, let $P_i=\{p\in P\colon\phi(p)=i\}$.
By Lemma \ref{lem:secondary-coloring}, each set $P_i$ can be further $2$-colored on-line so as to distinguish any $p,q\in P_i$ such that $pq\notin E(G)$ and there is some edge between $S(p)$ and $S(q)$.
Let $\zeta$ be such a $2$-coloring of each $P_i$ using colors $1$ and $2$.
For each $p\in P$, by Lemma \ref{lem:overlap-game-secondary}, $G[S(p)]$ can be properly colored on-line using colors $1,\ldots,\smash{\binom{k}{2}}$.
For $p\in P$ and $x\in S(p)$, let $\xi(x)$ denote the color of $x$ in such a coloring.
We color each vertex $x\in S(p)$ by the triple $(\phi(p),\zeta(p),\xi(x))$.
It follows that if $p,q\in P$, $x\in S(p)$, $y\in S(q)$, $(\phi(p),\zeta(p),\xi(x))=(\phi(q),\zeta(q),\xi(y))$, and $xy\in E(G)$, then $pq\in E(G)$.
Therefore, since $\phi$ is triangle-free, the coloring by triples is a triangle-free coloring of $G$ using at most $2k\smash{\binom{k}{2}}$ colors.
\end{proof}

Theorems \ref{thm:subtree} \ref{item:subtree-upper}--\ref{item:clean-subtree-upper} and \ref{thm:rectangle} now follow from Theorem \ref{thm:reduction-to-clean}, Lemmas \ref{lem:subtree-game} and \ref{lem:rectangle-to-game} (respectively), Lemma \ref{lem:heavy-paths} and the discussion that follows it, Lemmas \ref{lem:clean-subtree-coloring} and \ref{lem:overlap-game-coloring} (respectively), and Lemma \ref{lem:game-upper}.
Theorem \ref{thm:rectangle-K3} follows from Lemmas \ref{lem:rectangle-to-game}, \ref{lem:overlap-game-K3-coloring} and \ref{lem:game-upper}.

In the next section, we will prove that the proper coloring algorithm of clean subtree overlap graphs presented above uses the asymptotically optimal number of colors.

\section{Subtree overlap graphs with large chromatic number}
\label{sec:construction}

\def\present#1{{\tt present(}{\normalfont #1}{\tt)}}

In this final section, we will present a construction of clean subtree overlap graphs with chromatic number $\Theta_\omega((\log\log n)^{\omega-1})$ and thus prove Theorem \ref{thm:subtree}~\ref{item:clean-subtree-lower}.
To this end, we will prove the following.

\begin{lemma}
\label{lem:subtree-strategy}
For\/ $k,m\geq 1$, Presenter has a finite strategy to force Algorithm to use at least\/ $2m^{k-1}-1$ colors in\/ $2^{\smash{O_k(m)}}$ rounds of the game\/ $\gamABS(k)$.
Moreover, the number of presentation scenarios for all possible responses of Algorithm is\/ $2^{2^{\smash{O_k(m)}}}$.
\end{lemma}

We will generalize the strategy of Presenter forcing the use of $c$ colors in $2^{c-1}$ rounds of the game $\gamIOV(2)$, described in \cite{KPW15,PKK+13}.
The strategy that we will describe presents a set of vertices with relations $\includes$, $\overlaps$ and $\parallel$ that partition the order of presentation $\prec$ and satisfy the conditions \ref{item:A1}--\ref{item:A4}.
The graph $G$ is defined on these vertices by the relation $\overlaps$, that is, so that $xy\in E(G)$ if and only if $x\overlaps y$ or $y\overlaps x$.
The strategy ensures $\omega(G)\leq k$, so that all conditions of the definition of $\gamABS(k)$ are satisfied.

For convenience, we extend the notation $\includes$, $\overlaps$ and $\parallel$ to sets of vertices in a natural way.
For example, $X\includes Y$ denotes that $x\includes y$ for all $x\in X$ and $y\in Y$.
The strategy is expressed in terms of a recursive procedure {\tt present}, initially called as \present{$k$, $2m$, $m$, $\emptyset$, $\emptyset$}.
When a call to \present{$k$, $\ell$, $m$, $A_1$, $A_2$} occurs, the following context and conditions are assumed:
\begin{itemize}
\item some set of vertices, call it $P$, has been already presented in the game,
\item relations $\includes$, $\overlaps$ and $\parallel$ among the vertices in $P$ have been declared,
\item $A_1$ and $A_2$ are disjoint subsets of $P$ such that $A_1\includes A_2$,
\item $2\leq\ell\leq 2m$.
\end{itemize}
As a result of the call to \present{$k$, $\ell$, $m$, $A_1$, $A_2$} considered, the game is continued in such a way that the following happens:
\begin{itemize}
\item a new set of vertices is presented, call it $S$,
\item relations $\includes$, $\overlaps$ and $\parallel$ are declared between $P$ and $S$ in a fixed way---so that $A_1\includes S$, $A_2\overlaps S$, and $P\setminus(A_1\cup A_2)\parallel S$,
\item relations $\includes$, $\overlaps$ and $\parallel$ are declared among the vertices in $S$,
\item an independent set $R\subseteq S$ is picked so that Algorithm has used many (at least $\ell m^{k-2}-1$) colors on $R$.
\end{itemize}
Here is the procedure; it uses \textbf{return} statements to pass the set $R$ to the caller:

\begin{center}
\begin{algorithm}[H]
\BlankLine
\textbf{procedure \present{$k$, $\ell$, $m$, $A_1$, $A_2$}}\par
\uIf{$k=1$}{
  present a new vertex $y$ and declare that $x\includes y$ for every $x\in A_1$, $x\overlaps y$ for every $x\in A_2$, and $x\parallel y$ for every $x\notin A_1\cup A_2$ that has been presented before\;
  \Return $\{y\}$\;
}
\uElseIf{$\ell=2$}{
  \Return\ \present{$k-1$, $2m$, $m$, $A_1$, $A_2$}\;
}
\Else{
  $R_1\coloneq{}$\present{$k$, $\ell-1$, $m$, $A_1$, $A_2$}\;
  $R_2\coloneq{}$\present{$k$, $\ell-1$, $m$, $A_1\cup R_1$, $A_2$}\;
  \uIf{\upshape Algorithm has used at least $\ell m^{k-2}-1$ colors on $R_1\cup R_2$}{
    \Return $R_1\cup R_2$\;
  }
  \Else{
    $R_3\coloneq{}$\present{$k-1$, $2m$, $m$, $A_1\cup R_1$, $A_2\cup R_2$}\;
    \Return $R_1\cup R_3$\;
  }
}
\end{algorithm}
\end{center}

Before analyzing the details and proving correctness of the procedure, we first explain its core idea.
It lies in the case $k\geq 2$ and $3\leq\ell\leq 2m$ (the outer \textbf{else} block).
The requirements on the recursive calls to {\tt present} imply that $R_1$, $R_2$ and $R_3$ are independent sets, $R_1\includes R_2$, $R_1\includes R_3$, $R_2\overlaps R_3$, and appropriately many colors have been used on each of $R_1$, $R_2$ and $R_3$ (the last one considered only when used in the inner \textbf{else} block).
If the number of colors used on $R_1\cup R_2$ is large enough, then $R_1\cup R_2$ is a good candidate for the set $R$ to be returned by the current call to {\tt present}.
Otherwise, $R_1\cup R_3$ is a good candidate, as there are appropriately many common colors used on $R_1$ and $R_2$, and the colors used on $R_3$ must be different.
See Figure \ref{fig:present} for an illustration.

\begin{figure}[t]
\centering
\begin{tikzpicture}[xscale=.6,yscale=.35]
  \begin{scope}[line cap=rect]
  \tikzstyle{dashed}=[line cap=butt,dash pattern=on 0pt off 5pt on 3pt off 3pt on 3pt off 3pt on 3pt off 100pt];
  \clip (-5,0) rectangle (17.5,8.3);
  \draw[line width=4pt,color=gray,opacity=.5,rounded corners=10pt] (-4.5,0)--(-3,8)--(11,8);
  \draw[line width=4pt,color=gray,opacity=.5,dashed] (11,8)--(12.5,8);
  \draw[line width=4pt,color=gray,opacity=.5,rounded corners=10pt] (-3.5,0)--(-3,2)--(11,2);
  \draw[line width=4pt,color=gray,opacity=.5,dashed] (11,2)--(12.5,2);
  \draw[line width=4pt,color=red,opacity=.5,rounded corners=10pt] (1,0)--(2,4)[rounded corners=5pt]--(3.5,4)--(4.5,6)--(8.5,6)--(9.5,5)--(11,5);
  \draw[line width=4pt,color=red,opacity=.5,dashed] (11,5)--(12.5,5);
  \draw[line width=4pt,color=blue,opacity=.5,rounded corners=10pt] (6,0)--(7,4)--(11,4);
  \draw[line width=4pt,color=blue,opacity=.5,dashed] (11,4)--(12.5,4);
  \end{scope}
  \node[inner sep=1pt,fill=white] at (-2,8) {$A_1$};
  \node[inner sep=1pt,fill=white] at (-2,2) {$A_2$};
  \node[diamond,aspect=2,inner sep=0pt,fill=white] at (4,5) {$R_1$};
  \node[inner sep=1pt,fill=white] at (9,4) {$R_2$};
  \fill[gray,opacity=.5] (0,0) rectangle (3,5);
  \fill[gray,opacity=.5] (5,0) rectangle (8,5);
  \draw (-1,0)--(-1,7)--(10.5,7)--(10.5,0);
  \draw (0,0)--(0,5)--(3,5)--(3,0);
  \draw (5,0)--(5,5)--(8,5)--(8,0);
  \draw (-5,0)--(17.5,0);
  \begin{scope}[yshift=-7pt]
  \draw[decorate,decoration=brace] (3,0)--(0,0);
  \draw[decorate,decoration=brace] (8,0)--(5,0);
  \end{scope}
  \node[inner sep=1pt,anchor=west] (a) at (-5,-1.5) {\present{$k$, $\ell-1$, $m$, $A_1$, $A_2$}};
  \node[inner sep=1pt,anchor=west] (b) at (-5,-2.9) {\present{$k$, $\ell-1$, $m$, $A_1\cup R_1$, $A_2$}};
  \draw[dotted] (6.5,-22pt)--(6.5,-2.9)--(b.east);
\end{tikzpicture}\\[2ex]
\begin{tikzpicture}[xscale=.6,yscale=.375]
  \begin{scope}[line cap=rect]
  \tikzstyle{dashed}=[line cap=butt,dash pattern=on 0pt off 5pt on 3pt off 3pt on 3pt off 3pt on 3pt off 100pt];
  \clip (-5,0) rectangle (17.5,8.3);
  \draw[line width=4pt,color=gray,opacity=.5,rounded corners=10pt] (-4.5,0)--(-3,8)--(16,8);
  \draw[line width=4pt,color=gray,opacity=.5,dashed] (16,8)--(17.5,8);
  \draw[line width=4pt,color=gray,opacity=.5,rounded corners=10pt] (-3.5,0)--(-3,2)--(16,2);
  \draw[line width=4pt,color=gray,opacity=.5,dashed] (16,2)--(17.5,2);
  \draw[line width=4pt,color=red,opacity=.5,rounded corners=10pt] (1,0)--(2,4)[rounded corners=5pt]--(3.5,4)--(4.5,6)--(13.5,6)--(14.5,5)--(16,5);
  \draw[line width=4pt,color=red,opacity=.5,dashed] (16,5)--(17.5,5);
  \draw[line width=4pt,color=red,opacity=.5,rounded corners=10pt] (6,0)--(7,4){[rounded corners=5pt]--(8.5,4)--(9.5,3)}--(13.75,3)--(14.75,0);
  \draw[line width=4pt,color=blue,opacity=.5,rounded corners=10pt] (11,0)--(12,4)--(16,4);
  \draw[line width=4pt,color=blue,opacity=.5,dashed] (16,4)--(17.5,4);
  \end{scope}
  \node[inner sep=1pt,fill=white] at (-2,8) {$A_1$};
  \node[inner sep=1pt,fill=white] at (-2,2) {$A_2$};
  \node[diamond,aspect=2,inner sep=0pt,fill=white] at (4,5) {$R_1$};
  \node[diamond,aspect=.7,inner sep=0pt,fill=white] at (9,3.5) {$R_2$};
  \node[inner sep=1pt,fill=white] at (14,4) {$R_3$};
  \fill[gray,opacity=.5] (0,0) rectangle (3,5);
  \fill[gray,opacity=.5] (5,0) rectangle (8,5);
  \fill[gray,opacity=.5] (10,0) rectangle (13,5);
  \draw (-1,0)--(-1,7)--(15.5,7)--(15.5,0);
  \draw (0,0)--(0,5)--(3,5)--(3,0);
  \draw (5,0)--(5,5)--(8,5)--(8,0);
  \draw (10,0)--(10,5)--(13,5)--(13,0);
  \draw (-5,0)--(17.5,0);
  \begin{scope}[yshift=-7pt]
  \draw[decorate,decoration=brace] (3,0)--(0,0);
  \draw[decorate,decoration=brace] (8,0)--(5,0);
  \draw[decorate,decoration=brace] (13,0)--(10,0);
  \end{scope}
  \node[inner sep=1pt,anchor=west] (a) at (-5,-1.5) {\present{$k$, $\ell-1$, $m$, $A_1$, $A_2$}};
  \node[inner sep=1pt,anchor=west] (b) at (-5,-2.9) {\present{$k$, $\ell-1$, $m$, $A_1\cup R_1$, $A_2$}};
  \node[inner sep=1pt,anchor=west] (c) at (-5,-4.3) {\present{$k-1$, $2m$, $m$, $A_1\cup R_1$, $A_2\cup R_2$}};
  \draw[dotted] (6.5,-22pt)--(6.5,-2.9)--(b.east);
  \draw[dotted] (11.5,-22pt)--(11.5,-4.3)--(c.east);
\end{tikzpicture}
\captionsetup{singlelinecheck=off}
\caption[]{A recursion step of \present{$k$, $\ell$, $m$, $A_1$, $A_2$} with $k\geq 2$ and $3\leq\ell\leq 2m$ illustrated by how an interval filament model of the resulting graph might look in each of the two cases:
\begin{itemize}
\item top---at least $\ell m^{k-2}-1$ colors have been used on $R_1\cup R_2$,
\item bottom---at most $\ell m^{k-2}-2$ colors have been used on $R_1\cup R_2$.
\end{itemize}
Thick lines illustrate bundles of pairwise non-intersecting interval filaments, as labeled.
Each gray box covers the filaments of the vertices presented in a direct recursive call to {\tt present}; these filaments stay inside the box (unless they belong to the result set $R_i$) and cross all filaments piercing the box from left to right (e.g.\ those in $A_2$).
The white box plays the analogous role for the call to \present{$k$, $\ell$, $m$, $A_1$, $A_2$} itself.}
\label{fig:present}
\end{figure}
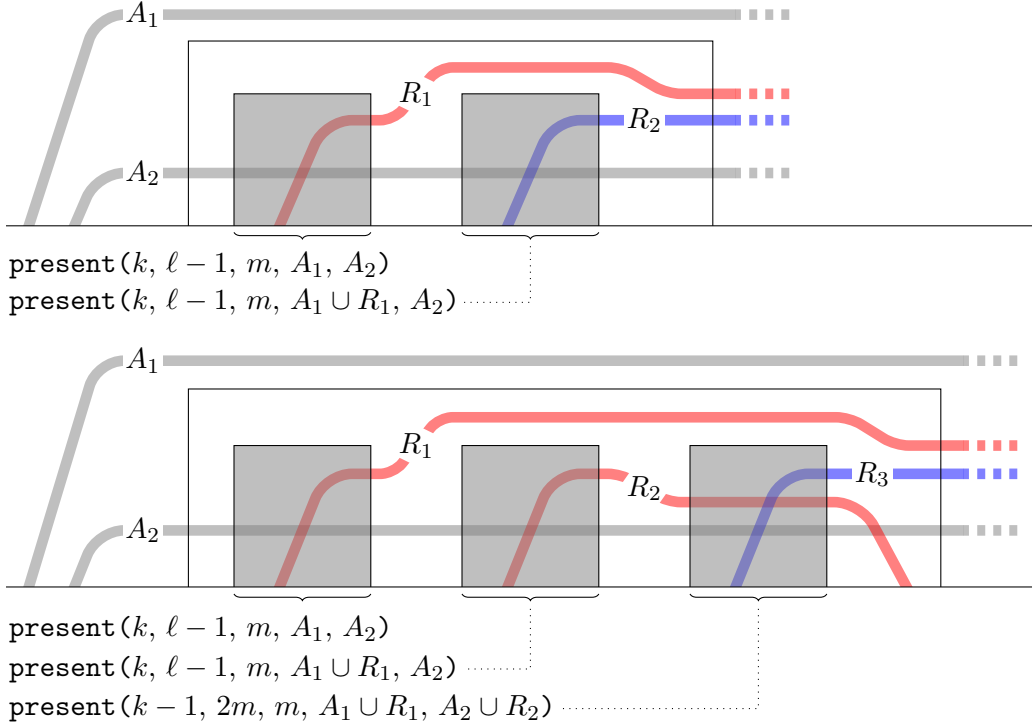

Our goal is to show that the initial call to \present{$k$, $2m$, $m$, $\emptyset$, $\emptyset$} yields a strategy that obeys the rules of the game $\gamABS(k)$ and forces the use of at least $2m^{k-1}-1$ colors.
This is achieved by the following lemma.

\begin{lemma}
\label{lem:strategy}
As a result of a call to\/ \present{$k$, $\ell$, $m$, $A_1$, $A_2$}, a set of vertices\/ $S$ is presented, relations\/ $\includes$, $\overlaps$ and\/ $\parallel$ are declared, and a set\/ $R\subseteq S$ is returned so that
\begin{enumerate}
\item\label{item:strat-A} $A_1\includes S$, $A_2\overlaps S$, and\/ $P\setminus(A_1\cup A_2)\parallel S$, where\/ $P$ denotes the set of vertices that have been presented before the call to\/ \present{$k$, $\ell$, $m$, $A_1$, $A_2$},
\item\label{item:strat-axioms} any\/ $x,y,z\in S$ satisfy the conditions \ref{item:A1}--\ref{item:A4} of the definition of\/ $\gamABS(k)$,
\item\label{item:strat-R} any\/ $x,y\in S$ satisfy the following conditions:
\begin{enumeratex}{B}
\item\label{item:R1} if\/ $x\prec y$ and\/ $x,y\in R$, then\/ $x\includes y$,
\item\label{item:R2} if\/ $x\includes y$ and\/ $y\in R$, then\/ $x\in R$,
\item\label{item:R3} if\/ $x\parallel y$, then\/ $x\in S\setminus R$,
\end{enumeratex}
\item\label{item:strat-omega} the graph defined on\/ $S$ by the relation\/ $\overlaps$ has clique number at most\/ $k$,
\item\label{item:strat-colors} Algorithm has used at least\/ $\ell m^{k-2}-1$ colors on the vertices in\/ $R$.
\end{enumerate}
\end{lemma}

\begin{proof}
The recursion tree of calls to {\tt present} is finite, because every call to {\tt present} with first parameter $k$ and second parameter $\ell$ makes recursive calls to {\tt present} with $k$ smaller by $1$ or with $k$ unchanged and $\ell$ smaller by $1$.
The proof of the lemma goes by induction on the recursion tree.
That is, we prove the lemma for a particular call to \present{$k$, $\ell$, $m$, $A_1$, $A_2$}, called the \emph{current call} henceforth, assuming that the lemma holds for every recursive call to {\tt present} triggered as a result of the current call.
Even though the current call occurs in the context of a set of vertices $P$ presented before the call, only part \ref{item:strat-A} of the statement of the lemma is concerned about $P$.
In particular, the lemma considered just for the current call does \emph{not} assert the conditions \ref{item:A1}--\ref{item:A4} for all triples of vertices in $P\cup S$; they are to be considered at higher levels of induction.

First, consider the case that the current call is a leaf of the recursion tree, which happens if and only if $k=1$.
In that case, \ref{item:strat-A} is a direct consequence of how the relations $\includes$, $\overlaps$ and $\parallel$ are declared when a new vertex is presented by the procedure, and \ref{item:strat-axioms}--\ref{item:strat-colors} hold trivially, because they are concerned about a set $S$ that consists of just one vertex.

Next, if $k\geq 2$ and $\ell=2$, then \ref{item:strat-A}--\ref{item:strat-colors} follow directly from the induction hypothesis applied to the recursive call with $k-1$ and $2m$ in place of $k$ and $\ell$.

Finally, for the rest of the proof, consider the case that $k\geq 2$ and $3\leq\ell\leq 2m$.
The current call yields two or three recursive calls to {\tt present}, as a result of which three sets of vertices $S_1$, $S_2$ and $S_3$ are presented, respectively (we let $S_3=\emptyset$ if there is no third recursive call).
Thus $A_1\cup A_2\prec S_1\prec S_2\prec S_3$ and $S=S_1\cup S_2\cup S_3$.
The induction hypothesis \ref{item:strat-A} applied to the recursive calls implies \ref{item:strat-A} for $S$ as well as the following:
\begin{equation*}
\label{eq:S}
R_1\includes S_2\cup S_3,\quad S_1\setminus R_1\parallel S_2\cup S_3,\quad R_2\overlaps S_3,\quad S_2\setminus R_2\parallel S_3.\tag{$*$}
\end{equation*}

To prove \ref{item:strat-axioms} for $S$, choose any $x,y,z\in S$ with $x\prec y\prec z$.
If $x,y,z\in S_i$, then all \ref{item:A1}--\ref{item:A4} follow directly from the induction hypothesis \ref{item:strat-axioms} for the recursive calls.
If $x\in S_i$ and $y\in S_j$ with $i<j$, then, by \eqref{eq:S}, the relation between $x$ and $y$ is the same as the relation between $x$ and $z$, whence all \ref{item:A1}--\ref{item:A4} follow.
It remains to consider the case that $x,y\in S_i$ and $z\in S_j$ with $i<j$.
To this end, we use \eqref{eq:S} and the induction hypothesis \ref{item:strat-R} applied to the recursive calls.
\begin{enumeratex}{A}
\item Suppose $x\includes y$ and $y\includes z$.
It follows from $y\includes z$ and \eqref{eq:S} that $y\in R_1$ and $z\in S_2\cup S_3$.
This and $x\includes y$ imply $x\in R_1$, by \ref{item:R2}.
Hence $x\includes z$, by \eqref{eq:S}.
\item Suppose $x\includes y$ and $y\overlaps z$.
It follows from $y\overlaps z$ and \eqref{eq:S} that $y\in R_2$ and $z\in S_3$.
This and $x\includes y$ imply $x\in R_2$, by \ref{item:R2}.
Hence $x\overlaps z$, by \eqref{eq:S}.
\item Suppose $x\overlaps y$ and $y\includes z$.
It follows from $y\includes z$ and \eqref{eq:S} that $y\in R_1$ and $z\in S_2\cup S_3$.
This and $x\overlaps y$ imply $x\in S_1\setminus R_1$, by \ref{item:R1}.
Hence $x\parallel z$, by \eqref{eq:S}.
\item If $x\parallel y$, then $x\in S_i\setminus R_i$, by \ref{item:R3}.
This and $z\in S_j$ with $i<j$ imply $x\parallel z$, by \eqref{eq:S}.
\end{enumeratex}

To prove \ref{item:strat-R} for $R$, choose any $x,y\in S$ with $x\prec y$.
If $x,y\in S_i$, then all \ref{item:R1}--\ref{item:R3} follow directly from the induction hypothesis \ref{item:strat-R} applied to the recursive calls and from the fact that $R\cap S_i=R_i$ or $R\cap S_i=\emptyset$.
It remains to consider the case that $x\in S_i$ and $y\in S_j$ with $i<j$.
To this end, we use \eqref{eq:S} and the fact that the procedure {\tt present} returns $R=R_1\cup R_2$ or $R=R_1\cup R_3$.
\begin{enumeratex}{B}
\item If $x,y\in R$, then $x\in R_1$ and $y\in R_j$, by the definition of $R$, so $x\includes y$, by \eqref{eq:S}.
\item If $x\includes y$, then $x\in R_1$, by \eqref{eq:S}, so $x\in R$, by the definition of $R$.
\item If $x\parallel y$, then $x\in S_i\setminus R_i$, by \eqref{eq:S}, so $x\in S\setminus R$, by the definition of $R$.
\end{enumeratex}

We have $\omega(G[S])=\max\{\omega(G[S_1]),\omega(G[S_2]),\omega(G[S_3])+1\}\leq k$, by \eqref{eq:S} and by the property \ref{item:R1} of $R_2$.
Hence we have \ref{item:strat-omega} for $S$.

Finally, we prove \ref{item:strat-colors} for $R$.
If Algorithm has used at least $\ell m^{k-2}-1$ colors on $R_1\cup R_2$, then the call returns $R=R_1\cup R_2$, so \ref{item:strat-colors} holds.
It remains to consider the opposite case---that at most $\ell m^{k-2}-2$ colors have been used on $R_1\cup R_2$ and the call returns $R=R_1\cup R_3$.
By \ref{item:strat-colors} applied to $R_1$ and $R_2$, Algorithm has used at least $(\ell-1)m^{k-2}-1$ colors on each of $R_1$ and $R_2$.
It follows that at least $(\ell-2)m^{k-2}$ common colors have been used on both $R_1$ and $R_2$.
By \ref{item:strat-colors} applied to $R_3$, Algorithm has used at least $2m^{k-2}-1$ colors on $R_3$.
Since $R_2\overlaps R_3$, these colors must be different from the common colors used on both $R_1$ and $R_2$.
Therefore, at least $\ell m^{k-2}-1$ colors have been used on $R_1\cup R_3$, which proves \ref{item:strat-colors} for $R$.
\end{proof}

\begin{proof}[Proof of Lemma \ref{lem:subtree-strategy}]
By Lemma \ref{lem:strategy} \ref{item:strat-axioms}, \ref{item:strat-omega} and \ref{item:strat-colors}, the strategy of Presenter described by a call to \present{$k$, $2m$, $m$, $\emptyset$, $\emptyset$} obeys the rules of the game $\gamABS(k)$ and forces Algorithm to use at least $2m^{k-1}-1$ colors in total.
It remains to prove that the number of presentation scenarios for all possible responses of Algorithm is $2^{2^{\smash{O_k(m)}}}$.

The only conditional instruction in the procedure {\tt present} whose result is not determined by the values of $k$, $\ell$ and $m$ but depends on the coloring chosen by Algorithm is the test whether ``Algorithm has used at least $\ell m^{k-2}-1$ colors on $R_1\cup R_2$''.
We call it simply a \emph{test}.

Let $s_{k,\ell}$ and $c_{k,\ell}$ denote the maximum number of vertices that can be presented and the maximum number of tests that can be performed, respectively, as a result of a call to \present{$k$, $\ell$, $m$, $A_1$, $A_2$}.
It easily follows from the procedure that
\begin{gather*}
\begin{alignedat}{4}
s_{1,\ell}&=1, &\qquad s_{k,2}&=s_{k-1,2m} &\quad &\text{for $k\geq 2$},\\
c_{1,\ell}&=0, &\qquad c_{k,2}&=c_{k-1,2m} &\quad &\text{for $k\geq 2$},
\end{alignedat}\\
\begin{alignedat}{2}
s_{k,\ell}&\leq 2s_{k,\ell-1}+s_{k-1,2m} &\quad &\text{for $k\geq 2$ and $3\leq\ell\leq 2m$},\\
c_{k,\ell}&\leq 2c_{k,\ell-1}+c_{k-1,2m}+1 &\quad &\text{for $k\geq 2$ and $3\leq\ell\leq 2m$}.
\end{alignedat}
\end{gather*}
This yields the following by straightforward induction:
\begin{gather*}
\begin{alignedat}{3}
s_{1,2m}&=1, &\qquad s_{k,2m}&\leq(2^{2m-1}-1)s_{k-1,2m} &\quad &\text{for $k\geq 2$},\\
c_{1,2m}&=0, &\qquad c_{k,2m}&\leq(2^{2m-1}-1)(c_{k-1,2m}+1)-1 &\quad &\text{for $k\geq 2$},
\end{alignedat}\\
\begin{aligned}
s_{k,2m}&\leq(2^{2m-1}-1)^{k-1},\\
c_{k,2m}&\leq(2^{2m-1}-1)^{k-1}-1.
\end{aligned}
\end{gather*}

For fixed $k$ and $m$, the outcome of the call to \present{$k$, $2m$, $m$, $\emptyset$, $\emptyset$} (that is, the sequence of vertices presented and the relations $\includes$, $\overlaps$ and $\parallel$ declared) is entirely determined by the outcomes of the tests performed as a result of that call.
Since at most $c_{k,2m}$ tests are performed, the number of possible outcomes of the call is at most $2^{c_{k,2m}}$.
Since at most $s_{k,2m}$ vertices are presented, each outcome of the call gives rise to at most $s_{k,2m}$ presentation scenarios, each corresponding to an initial segment of the sequence of vertices presented.
Therefore, the total number of presentation scenarios possible with this strategy is at most $2^{c_{k,2m}}s_{k,2m}$, which is $2^{2^{\smash{O_k(m)}}}$.
\end{proof}

Lemmas \ref{lem:subtree-strategy}, \ref{lem:game-lower} and \ref{lem:subtree-game} yield a construction of clean subtree overlap graphs (string graphs) with $\chi=\Theta_\omega((\log\log n)^{\omega-1})$.
This completes the proof of Theorem \ref{thm:subtree}~\ref{item:clean-subtree-lower}.

The graphs constructed above also satisfy $\chi_\omega=\Theta_\omega(\log\log n)$, because every color class of a $K_{\omega(G)}$-free coloring of a clean subtree overlap graph $G$ induces a subgraph with chromatic number $O_{\omega(G)}((\log\log n)^{\smash[t]{\omega(G)}-2})$, by Theorem \ref{thm:subtree}~\ref{item:clean-subtree-upper}.
All intersection models of the graphs constructed above require that some pairs of curves intersect many times.
This is because these graphs contain vertices whose neighborhoods have chromatic number $\Theta_\omega((\log\log n)^{\omega-2})$, while the neighborhood of every vertex of an intersection graph of $1$-intersecting curves (that is, curves any two of which intersect in at most one point) has bounded chromatic number \cite{SW15}.
We wonder whether it is possible to construct intersection graphs of $1$-intersecting curves with bounded clique number and with chromatic number asymptotically greater than $\log\log n$ (that $\chi=\Theta(\log\log n)$ can be achieved follows from the results of \cite{PKK+13,PKK+14}).

\section*{Acknowledgments}

We thank Martin Pergel for familiarizing us with subtree overlap graphs and interval filament graphs and for asking whether these classes of graphs are $\chi$-bounded.
We also thank anonymous reviewers for their very helpful corrections and comments.

\end{document}